\documentclass[a4paper,10pt]{article}
\usepackage[utf8]{inputenc}
\usepackage{fancyhdr}
\usepackage{graphicx}
\usepackage{amsmath}
\usepackage{amssymb}
\usepackage{hyperref}
\usepackage[all]{hypcap}
\usepackage{pdfpages}
\usepackage{mathrsfs}
\usepackage{caption}
\usepackage{subcaption}
\usepackage{bm}
\usepackage{enumitem}
\usepackage{algorithm}
\usepackage{algorithmicx}
 \usepackage{algpseudocode}
\usepackage[margin=1in]{geometry}
\usepackage{tikz}
\usetikzlibrary{shapes.geometric, arrows,calc}
\newtheorem{theorem}{Theorem}[section]
\newtheorem{lemma}[theorem]{Lemma}

\newtheorem{corollary}[theorem]{Corollary}

\newenvironment{proof}[1][Proof]{\begin{trivlist}
\item[\hskip \labelsep {\bfseries #1}]}{\end{trivlist}}
\newenvironment{definition}[1][Definition]{\begin{trivlist}
\item[\hskip \labelsep {\bfseries #1}]}{\end{trivlist}}

\newenvironment{remark}[1][Remark]{\begin{trivlist}
\item[\hskip \labelsep {\bfseries #1}]}{\end{trivlist}}

\newcommand{\qed}{\nobreak \ifvmode \relax \else
      \ifdim\lastskip<1.5em \hskip-\lastskip
      \hskip1.5em plus0em minus0.5em \fi \nobreak
      \vrule height0.75em width0.5em depth0.25em\fi}
      
\tikzstyle{startstop} = [rectangle, rounded corners, minimum width=3cm, minimum height=1cm,text centered, draw=black, fill=red!30]
\tikzstyle{io} = [trapezium, trapezium left angle=70, trapezium right angle=110, minimum width=3cm, minimum height=1cm, text centered,draw=black, fill=blue!30]
\tikzstyle{process} = [rectangle, minimum width=3cm, minimum height=1cm, text centered, draw=black, fill=orange!30]
\tikzstyle{decision} = [diamond, minimum width=3cm, minimum height=1cm, draw=black, fill=green!30]
\tikzstyle{arrow} = [thick,->,>=stealth]
\title{Analysis of stochastic approximation schemes with set-valued maps in the absence of a stability guarantee and their stabilization}
\author{Vinayaka G. Yaji and Shalabh Bhatnagar,\\ Department of Computer Science and Automation,
        \\ Indian Institute of Science, Bangalore.\\
        vgyaji@gmail.com, shalabh@csa.iisc.ernet.in}
\begin{document}

\maketitle
\begin{abstract}
In this paper, we analyze the behavior of stochastic approximation schemes with set-valued maps in the absence of a stability guarantee. We prove that after a large number of iterations if the stochastic approximation process enters the domain of attraction of an attracting set it gets locked into the attracting set with high probability. We demonstrate that the above result is an effective instrument for analyzing stochastic approximation schemes in the absence of a stability guarantee, by using it obtain an alternate criteria for convergence in the presence of a locally attracting set for the mean field and by using it to show that a feedback mechanism, which involves resetting the iterates at regular time intervals, stabilizes the scheme when the mean field possesses a globally attracting set, thereby guaranteeing convergence. The results in this paper build on the works of \it{V.S. Borkar, C. Andrieu }\rm and \it{H. F. Chen }\rm, by allowing for the presence of set-valued drift functions.   
\end{abstract}

\section{Introduction}
\label{intro}
It is well known that several optimization and control tasks can be cast as a root finding problem. That is, given $f:\mathbb{R}^d\rightarrow\mathbb{R}^d$, one needs to find $x^*\in\mathbb{R}^d$, such that $f(x^*)=0$ (given such a point exists). Due to practical considerations, one usually has access to noisy measurements/estimations of the function whose root needs to be determined. An approach to solving such a problem with noisy measurements of $f$, is given by the recursion,
\begin{equation}
\label{std_rec}
X_{n+1}-X_{n}-a(n)M_{n+1}=a(n)f(X_n),
\end{equation}
where  $\{M_n\}_{n\geq1}$, denotes the noise arising in the measurement of $f$ and having fixed an initial condition ($X_0\in\mathbb{R}^d$), the iterates $\{X_n\}_{n\geq1}$ are generated according to recursion $\eqref{std_rec}$. \cite{benaim3} under certain assumptions which include the Lipschitz continuity of the function $f$, boundedness of the iterates along almost every sample path (that is $\mathbb{P}(\sup_{n\geq0}\|X_n\|<\infty)=1$) and a condition which ensures that the eventual contribution of the additive noise terms is negligible, showed that the linearly interpolated trajectory of recursion $\eqref{std_rec}$ tracks the flow of the ordinary differential equation (o.d.e.) given by,
\begin{equation}
\label{std_ode}
\frac{dx}{dt}=f(x).
\end{equation}
Such a trajectory is called an \it{asymptotic pseudotrajectory }\rm for the flow of  o.d.e.$\eqref{std_ode}$ (for a precise definition see \cite{benaim3}). Suppose the set of zeros of $f$ is a globally asymptotically stable set for the flow of o.d.e. $\eqref{std_ode}$, then it was shown that the limit set of an asymptotic pseudotrajectory was contained in such a set and hence the iterates $\{X_n\}_{n\geq0}$ converge in the limit to a root of the function $f$.

 In order to analyze recursion $\eqref{std_rec}$ when the function $f$ is no longer Lipschitz continuous or even continuous, but is just measurable satisfying the linear growth property, that is for every $x\in \mathbb{R}^d$, $\|f(x)\|\leq K(1+\|x\|)$ for some $K>0$, or when there is a non-additive noise/control component taking values in a compact set whose law is not known 
(in which case the recursion $\eqref{std_rec}$ takes the form $X_{n+1}-X_{n}-a(n)M_{n+1}=a(n)f(X_n,U_n)$, where $U_n$ denotes the noise/control), the above mentioned o.d.e. method needed to be extended to recursions with much weaker requirements on the function $f$. This was accomplished in \cite{benaim1}, where the asymptotic behavior of the recursion given by,
\begin{equation}
\label{std_rec_1}
X_{n+1}-X_{n}-a(n)M_{n+1}\in a(n)F(X_n),
\end{equation}
was studied, where $F$ is a set-valued map satisfying some conditions (while the other quantities have same interpretation as in $\eqref{std_rec}$). Under the assumption of stability of iterates (that is $\mathbb{P}(\sup_{n\geq0}\|X_n\|<\infty)=1$) and appropriate conditions on the additive noise terms, in \cite{benaim1}, it was shown that the linearly interpolated trajectory of recursion $\eqref{std_rec_1}$ tracks the flow of the differential inclusion (d.i.) given by, 
\begin{equation}
\label{std_di}
\frac{dx}{dt}\in F(x).
\end{equation}
We refer the reader to \cite[Ch.~5.3]{borkartxt} for a detailed argument as to how the measurable case and the case with unknown noise/control be recast in the form of recursion $\eqref{std_rec_1}$. For a brief summary of the convergence analysis of recursion $\eqref{std_rec_1}$ we refer the reader to section \ref{mtvtn} of this paper. 

Common to the analysis of both recursion $\eqref{std_rec}$ and $\eqref{std_rec_1}$ is the assumption on the stability of the iterates, that is $\mathbb{P}(\sup_{n\geq0}\|X_n\|<\infty)=1$. The condition of stability is highly non-trivial and difficult to verify. Over the years significant effort has gone into providing sufficient conditions for stability (see \cite{borkarmeyn}, \cite{arunstab}). In \cite{borkarlkp}, it was shown that for recursion $\eqref{std_rec}$, in the absence of stability guarantee, the probability of converging to an attracting set of o.d.e. $\eqref{std_ode}$ given that the iterates lie in a neighborhood of it converged to one as the index ($n$) in which the iterate entered the neighborhood of the attracting set increased to infinity. This probability of the iterates converging to an attracting set given that the iterate lies in a neighborhood of it is called the \it{lock-in probability }\rm and in \cite{borkarlkp} a lower bound for the same was used to obtain sample complexity bounds for recursion $\eqref{std_rec}$. Further a tighter lower bound for the lock-in probability was derived in \cite{kamal} under a slightly stronger noise assumption and used to obtain convergence guarantee when the law of the iterates are tight. In this paper we extend the results in \cite{borkarlkp} to the case of stochastic approximation schemes with set-valued maps as in recursion $\eqref{std_rec_1}$.

\subsection{Contributions and organization of the paper}   
We first provide a lower bound for the lock-in probability of stochastic approximation schemes with set-valued maps as in recursion $\eqref{std_rec_1}$. The bound is derived under an assumption on the additive noise terms which is stronger than the corresponding in \cite{borkarlkp}, which is necessitated due to the lack of Lipschitz continuity of the drift function $F$. We establish that,
\begin{equation*}
\mathbb{P}\left(X_n\to A\ \mathrm{as}\ n\to\infty|X_{n_0}\in\mathcal{O}'\right)\geq 1-2de^{-\tilde{K}/b(n_0)},
\end{equation*}
for $n_0$ large, where, $A\subseteq\mathbb{R}^d$, denotes an attracting set of DI \ref{std_di}, $\mathcal{O}'$ is an open neighborhood of $A$ with compact closure, $\tilde{K}$ is some positive constant and $\{b(n)\}_{n\geq0}$ is a sequence of reals converging to zero, which are step size dependent. 

Having summarized the convergence analysis under stability in section \ref{mtvtn}, we state the lock-in probability bound in section \ref{mr} and provide a few implications of the same. Using the lock-in probability result we provide an alternate criteria for convergence in the presence of a locally attracting set which removes the need to verify stability. A detailed comparison between the obtained convergence guarantee and the corresponding in the presence of stability is also provided. 

Proof of the lock-in probability result is presented in section \ref{prf_lip}. The proof relies heavily on the insights obtained from the analysis in \cite{borkarlkp} for single-valued maps. From the analysis in \cite{borkarlkp}, it is evident that the Lipschitz continuity of the drift function $f$ plays a crucial role in obtaining events and decoupling error contributions which in turn are necessary to obtain the bound in the inequality above. But in the recursion studied in this paper (that is recursion \eqref{std_rec_1}), the drift function $F$ is set-valued and the assumptions under which we study the said recursion (which are summarized in section \ref{recass}), the drift function $F$ is not even continuous. We overcome this problem by first obtaining a sequence of locally Lipschitz continuous set-valued maps which approximate the drift function $F$ from above and then parameterizing them using the Stiener selection procedure. The associated results are summarized in section \ref{usc_app}. This enables us to write recursion $\eqref{std_rec_1}$ in the form of recursion $\eqref{std_rec}$, but with locally Lipschitz continuous drift functions. Further the relation between the solutions of differential inclusions with the approximating set-valued maps as their vector field and those of DI $\eqref{std_di}$, is established in section \ref{sol_apprx}. Having written recursion $\eqref{std_rec_1}$ in the form of recursion $\eqref{std_rec}$, we then collect sample paths of interest in section \ref{prf_lip3}. Along the sample paths that are collected the iterates are such that, having entered a neighborhood of the attracting set at iteration $n_0$, the iterates will infinitely often enter the said neighborhood and the time elapsed between successive visits to the neighborhood of the attracting set can be upper bounded by a constant which is mean field dependent. Further we show that the probability of occurrence of such sample paths can be lower bounded by error contributions due to additive noise terms alone after a large number of iterations. Using the concentration inequality for martingale sequences we obtain the lock-in probability bound in section \ref{nbd}. 

Using the lock-in probability result we design a feedback mechanism which enables us to stabilize the stochastic approximation scheme in the presence of a globally attracting set for DI $\eqref{std_di}$. The feedback mechanism involves resetting the iterates at regular time intervals if they are found to be lying outside a certain compact set. This approach to stabilization has been studied in various forms for stochastic approximation schemes with single-valued drift functions as in recursion $\eqref{std_rec}$, in \cite{hfc_0}, \cite{hfc_1}, \cite{matti} and \cite{andvih} to name a few. We extend the same to the case of set-valued drift functions. The main idea in the analysis of such a scheme is to show that along almost every sample path of the modified recursion, the number of resets that are performed is finite, thereby guaranteeing that eventually the iterates lie within a compact set. We observe that the lock-in probability result (to be precise the approach adopted to obtain the lock-in probability result) plays a central role in showing that the number of resets performed remain finite. Having shown that the iterates eventually lie within a compact set, we use the convergence arguments from \cite{benaim1} to argue that the iterates generated by the modified scheme converge to the globally attracting set of DI $\eqref{std_di}$. The modified scheme is presented and explained in detail in section \ref{appl}. The proof of the finite resets theorem is presented in section \ref{prf_fnrst}. The procedure employed to collect sample paths in the proof of the lock-in probability result can be used to collect sample paths where only finite number of resets have occurred in the modified scheme and this in turn enables  us show that the number of resets are finite almost surely.

Finally, we conclude by providing a few directions for future work in section \ref{concl}.
\section{Recursion and assumptions}
\label{recass}
Let $\left(\Omega,\mathscr{F},\mathbb{P}\right)$ be a probability space and $\left\{X_n\right\}_{n\geq0}$ be a sequence of $\mathbb{R}^{d}$-valued 
random variables on $\Omega$, such that for every $n\geq0$,
\begin{equation}
\label{rec}
 X_{n+1}-X_{n}-a(n)M_{n+1}\in a(n)F(X_n),
\end{equation}
where, 
\begin{itemize}
 \item [(A1)] $F:\mathbb{R}^d\rightarrow\left\{\text{subsets of }\mathbb{R}^d\right\}$ is a set-valued map which for every $x\in\mathbb{R}^d$ 
 satisfies the following:
 \begin{itemize}
 \item [(i)] $F(x)$ is a convex and compact subset of $\mathbb{R}^d$,
 \item [(ii)] there exists $K>0$ (independent of $x$) such that $\sup_{y\in F(x)}\left\|y\right\|\leq K(1+\left\|x\right\|)$,
 \item [(iii)] for every $\mathbb{R}^d$-valued sequence $\{x_n\}_{n\geq1}$ converging to $x$ and for every sequence $\{y_{n}\in F(x_n)\}_{n\geq1}$ 
 converging to $y\in\mathbb{R}^d$, we have that $y\in F(x)$.
 \end{itemize}
 \item [(A2)] $\{a(n)\}_{n\geq0}$ is a sequence of positive real numbers satisfying,
 \begin{itemize}
 \item [(i)] $\sum_{n=0}^{\infty}a(n)=\infty$,
 \item [(ii)] $\sum_{n=0}^{\infty}\left(a(n)\right)^2<\infty$.
 \end{itemize}
 \item [(A3)] $\{M_n\}_{n\geq1}$ is a $\mathbb{R}^d$-valued, martingale difference sequence with respect to the filtration $\{\mathscr{F}_n:=\sigma(X_m,M_m,\ m\leq n)\}$. Furthermore, $\{M_n\}_{n\geq1}$ are  such that, 
 \begin{equation*}
\|M_{n+1}\|\leq K(1+\|x_n\|)\ a.s., 
 \end{equation*}
 for every $n\geq0$, for some constant $K>0$.
\end{itemize}

Assumption $(A1)$ ensures that the set-valued map $F$ is a Marchaud map. The condition $(A1)$(ii) is called the \it{linear growth property }\rm since it ensures that the size of the sets $F(x)$ grow linearly with respect to the distance from the origin. The condition $(A1)$(iii) is called the \it{closed graph property }\rm since it states that the graph of the set-valued map $F$, defined as, 
\begin{equation*}
\left\{(x,y)\in\mathbb{R}^{2d}:x\in\mathbb{R}^{d},\ y\in F(x)\right\}, 
\end{equation*}
is a closed subset of $\mathbb{R}^{2d}$. The map $F$ being a Marchaud map ensures that the differential inclusion (DI) given by,
\begin{equation}
\label{mfld}
\frac{dx}{dt}\in F(x),
\end{equation}
possesses at least one solution through every initial condition. By a solution of DI $\eqref{mfld}$ with initial condition $x_0\in\mathbb{R}^d$, we mean an absolutely continuous function $\bm{\mathrm{x}}:\mathbb{R}\rightarrow\mathbb{R}^{d}$ such that $\bm{\mathrm{x}}(0)=x_0$ and for almost every $t\in\mathbb{R}$, $\frac{d\bm{\mathrm{x}}(t)}{dt}\in F(\bm{\mathrm{x}}(t))$. DI $\eqref{mfld}$ is the \it{mean field }\rm of recursion $\eqref{rec}$ and its dynamics play an important role in describing the asymptotic behavior of recursion $\eqref{rec}$. 

Assumption $(A2)$ states the conditions to be satisfied by the step size sequence $\{a(n)\}_{n\geq0}$. Square summability (that is $(A2)$(ii)) is needed later in the analysis for obtaining a probability bound on certain tail events associated with the additive noise terms $\{M_n\}_{n\geq1}$. 

 Assumption $(A3)$,  defines the \it{martingale noise model}\rm. These terms denote the noise arising in the measurement of $F(\cdot)$. This condition holds in several reinforcement learning applications (see \cite[Ch.~10]{borkartxt})
 \begin{remark}
Clearly when $\{M_n\}_{n\geq1}$ are i.i.d. zero mean and bounded, assumption $(A3)$ is satisfied. Further, since the drift function in recursion $\eqref{rec}$ is a set-valued map, scenarios where the measurement noise terms possess a bounded bias can be recast in the form of recursion $\eqref{rec}$ as explained below.

Consider the recursion given by,
\begin{equation}
\label{apd}
 X_{n+1}-X_{n}-a(n)M_{n+1}-a(n)\eta_{n+1}=a(n) f(X_n),\ n\geq0,
\end{equation}
where $f:\mathbb{R}^d\rightarrow\mathbb{R}^d$ is a single-valued Lipschitz continuous map, for every $n\geq0$, $\eta_{n+1}$ denotes the bias in the measurement noise. Let the bias terms $\{\eta_n\}_{n\geq1}$ be bounded by a positive constant, say $\epsilon>0$ (that is, for every $n\geq1$, $\|\eta_n\|\leq\epsilon$). Then, recursion $\eqref{apd}$ can be written in the form of recursion $\eqref{rec}$ with set-valued map $F$, given by, $F(x)=\{f(x)+\eta: \|\eta\|\leq\epsilon\}$, for every $x\in\mathbb{R}^d$. We refer the reader to \cite[ch.~5.3]{borkartxt} for several other variants of the standard stochastic approximation scheme which can be analyzed with the help of recursion $\eqref{rec}$.
 \end{remark}

\section{Lock-in probability for stochastic recursive inclusions}

In order to state the main result of this paper, definition of the flow of  a DI,  an attracting set for such a dynamical system are needed. We recall these notions below and we state them with respect to the mean field of recursion $\eqref{rec}$ (for a detailed description and associated results see \cite{benaim1}).

The  \it{flow }\rm of DI $\eqref{mfld}$ is given by the set-valued map $\Phi:\mathbb{R}\times\mathbb{R}^d\rightarrow\{\text{subsets of }\mathbb{R}^d\}$, where for every $(t,x)\in\mathbb{R}\times\mathbb{R}^d$, 
\begin{equation}
\Phi(t,x):=\left\{\bm{\mathrm{x}}(t)\in\mathbb{R}^d:\bm{\mathrm{x}}(\cdot)\text{ is a solution of DI }\eqref{mfld}\ \text{with }\bm{\mathrm{x}}(0)=x\right\}.
\end{equation}

A compact set $A\subset\mathbb{R}^d$ is an \it{attracting set }\rm for the flow of DI $\eqref{mfld}$, if there exists an open  neighborhood of $A$, say $\mathcal{O}$, with the property that for every $\epsilon>0$, there exists a time $T>0$ (depending on $\epsilon$ and $\mathcal{O}$) such that for every $t\geq T$ and for every $x\in U$, $\Phi(t,x)\in N^{\epsilon}(A)$, where $N^{\epsilon}(A)$ denotes the $\epsilon$-neighborhood of $A$. Such a neighborhood $\mathcal{O}$ of an attracting set $A$ is called the \it{fundamental neighborhood }\rm of $A$. 

The set of initial conditions in $\mathbb{R}^d$ from which the flow is attracted to an attracting set $A$ is called the \it{basin of attraction }\rm and is denoted by $B(A)$. Formally, 
\begin{equation*}
B(A):=\left\{x\in\mathbb{R}^d: \cap_{t\geq0}\overline{\left\{\Phi(q,x):q\geq t\right\}}\subseteq A\right\}.
\end{equation*}
An attracting set $A$ is said to be \it{globally attracting }\rm if, $B(A)=\mathbb{R}^d$. 

%--------------------------------------------------------------------------------------------------------------------------------------------------------------------------------------------------------------------------------------
\subsection{Summary of the asymptotic analysis under stability}
\label{mtvtn}

Let $t(0):=0$ and for every $n\geq1$, $t(n):=\sum_{k=0}^{n-1}a(k)$. The linearly interpolated trajectory of recursion $\eqref{rec}$, is given by the stochastic process $\bar{X}:\Omega\times\mathbb{R}\rightarrow\mathbb{R}^d$, where for every $(\omega,t)\in\Omega\times[0,\infty)$,
\begin{equation}
\label{lit}
\bar{X}(\omega,t):=\left(\frac{t-t(n)}{t(n+1)-t(n)}\right)X_{n+1}(\omega)+\left(\frac{t(n+1)-t}{t(n+1)-t(n)}\right)X_n(\omega),
\end{equation}
where $n$ is such that $t\in[t(n),t(n+1))$ and for every $(\omega,t)\in\Omega\times(-\infty,0)$, $\bar{X}(\omega,t):=X_0(\omega)$. 

For $\omega\in\Omega$, the limit set map of $\bar{X}$ is given by, $\lambda:\Omega\rightarrow\{\text{subsets of }\mathbb{R}^d\}$ where for every $\omega\in\Omega$,
\begin{equation}
\label{ls}
\lambda(\omega):=\cap_{t\geq0}\overline{\left\{\bar{X}(\omega,q):q\geq t\right\}}.
\end{equation}
In \cite{benaim1}, under assumptions $(A1)-(A3)$ along with the additional assumption of stability of the iterates (that is $\mathbb{P}(\sup_{n\geq0}\|X_n\|<\infty)=1$), it was shown that for almost every $\omega\in\Omega$, the linearly interpolated trajectory of recursion $\eqref{rec}$, $\bar{X}(\omega,\cdot)$, is an \it{asymptotic pseudotrajectory }\rm for the flow of DI $\eqref{mfld}$. More precisely, for almost every $\omega\in\Omega$, $\bar{X}(\omega,\cdot)$ was shown to satisfy the following:
\begin{itemize}
\item [(a)] The family of shifted trajectories given by $\{\bar{X}(\omega,\cdot+t)\}_{t\geq0}$ is relatively compact in $\mathcal{C}(\mathbb{R},\mathbb{R}^d)$ where $\mathcal{C}(\mathbb{R},\mathbb{R}^d)$ denotes the metric space  of all continuous functions on $\mathbb{R}$ taking values in $\mathbb{R}^d$ with  metric $\bm{\mathrm{D}}$, which for every $\bm{\mathrm{z}},\ \bm{\mathrm{z}}'\in\mathcal{C}(\mathbb{R},\mathbb{R}^d)$ is given by
\begin{equation}
\label{mtrc}
\bm{\mathrm{D}}(\bm{\mathrm{z}},\bm{\mathrm{z}}')=\sum_{k=1}^{\infty}\frac{1}{2^k}\min\{\|\bm{\mathrm{z}}-\bm{\mathrm{z}}'\|_{[-k,k]},1\},
\end{equation}
where $\|\bm{\mathrm{z}}-\bm{\mathrm{z}}'\|_{[-k,k]}:=\sup_{t\in[-k,k]}\|\bm{\mathrm{z}}(t)-\bm{\mathrm{z}}'(t)\|$.
\item [(b)] Every limit point of the shifted trajectories $\{\bar{X}(\omega,\cdot+t)\}_{t\geq0}$ is a solution of the DI $\eqref{mfld}$.
\end{itemize}
From \cite[Thm.~4.3]{benaim1}, it follows that for almost every $\omega\in\Omega$, the limit set of the linearly interpolated trajectory $\bar{X}(\omega,\cdot)$, $\lambda(\omega)$, is a non-empty, compact and an internally chain transitive (ICT) set for the flow of DI $\eqref{mfld}$ (see \cite[Defn.~VI]{benaim1} for definition of an ICT set). Now using \cite[Thm.~3.23]{benaim1} the main convergence result of \cite{benaim1} follows and is stated below.
\begin{theorem}
\label{aptls}
Let $A\subseteq\mathbb{R}^d$ be an attracting set for the flow of DI $\eqref{mfld}$.
Under assumptions $(A1)-(A3)$, 
\begin{itemize}
\item [(a)] for almost every $\omega\in\{\omega\in\Omega:\sup_{n\geq0}\|X_n(\omega)\|<\infty\}\cap\{\omega\in\Omega:\lambda(\omega)\cap B(A)\neq\emptyset\}$, $\lambda(\omega)\subseteq A$ and therefore as $n\to\infty$, $X_n(\omega)\to A$.
\item [(b)] if $B(A)=\mathbb{R}^d$ (that is $A$ is a globally attracting set), then for almost every $\omega\in \{\omega\in\Omega:\sup_{n\geq0}\|X_n(\omega)\|<\infty\}$, $\lambda(\omega)\subseteq A$ and therefore as $n\to\infty$, $X_n(\omega)\to A$. 
\end{itemize}
\end{theorem}

The assumption of stability of the iterates used to obtain the above convergence result is highly non-trivial and difficult to verify. Moreover the proof method used to prove the above convergence result cannot be modified in a straight forward manner to obtain a similar convergence guarantee. This warrants an alternate approach to study the behavior of recursion $\eqref{rec}$ in the absence of stability guarantee and we accomplish this by extending the lock-in probability result from \cite{borkarlkp} to the set-valued case. Using the obtained lock-in probability bound we recover convergence guarantee similar to Theorem \ref{aptls} while eliminating the need to verify stability.

%---------------------------------------------------------------------------------------------------------------------------------------------------------------------------------------------------------------------------------------
\subsection{Main result and its implications}
\label{mr} 

Before we state the main result, we state an assumption which fixes the attracting set of interest.
\begin{itemize}
\item [(A4)] Let $A\subseteq \mathbb{R}^d$, be an attracting set of DI $\eqref{mfld}$ (the mean field of recursion $\eqref{rec}$) with $\mathcal{O}\subseteq\mathbb{R}^d$ as its fundamental neighborhood of attraction. 
\end{itemize}

Let $\mathcal{O}'$ be an open neighborhood of the attracting set $A$ (as in $(A4)$) such that $\bar{\mathcal{O}'}$ is compact and $\bar{\mathcal{O}'}\subseteq \mathcal{O}$. Then the main result of the paper can be stated as follows.

\begin{theorem}
\label{mnres}
\emph{(Lock-in probability)}
Under assumptions $(A1)-(A4)$, there exists a constant $\tilde{K}>0$ (depending on the attracting set $A$ and $\mathcal{O}'$) and an $N_0\geq1$ such that, for every $n_0\geq N_0$, for every $E\in \mathscr{F}_{n_0}$ satisfying $E\subseteq\{\omega\in\Omega:X_{n_0}(\omega)\in\mathcal{O}'\}$ and $\mathbb{P}(E)>0$, we have that,
\begin{equation*}
\mathbb{P}(X_n\to A\ \text{as }n\to\infty| E)\geq 1-2de^{-\tilde{K}/b(n_0)},
\end{equation*}
where, for every $n\geq0$, $b(n):=\sum_{k=n}^{\infty}(a(k))^2$.
\end{theorem}
There are two immediate implications of the above result and are stated below, one of which serves as an alternate convergence result in the absence of stability guarantees, that is it allows us to obtain the convergence guarantee in Theorem \ref{aptls}$(a)$ without the need to verify whether a given sample path satisfies $\sup_{n\geq0}\|X_n(\omega)\|<\infty$.
\begin{itemize}
\item [(1)] As a consequence of assumption $(A2)$(ii), we have that $\lim_{n\to\infty}b(n)=0$. Therefore from Theorem \ref{mnres}, if the observation that iterate lies in a neighborhood of the attracting set is made later in time ($n_0$), the probability of converging to the attracting set increases and converges to one as $n_0\to\infty$. Formally,
\begin{equation*}
\lim_{n_0\to\infty}\mathbb{P}(X_n\to A\ \text{as }n\to\infty|X_{n_0}\in\mathcal{O}')=1.
\end{equation*}
\item [(2)] Suppose $\mathbb{P}(\cap_{N\geq0}\cup_{n\geq N}\{X_n\in\mathcal{O}'\})>0$ (if $\mathbb{P}(\cap_{N\geq0}\cup_{n\geq N}\{X_n\in\mathcal{O}'\})=0$ then the iterates almost surely do not converge to the attracting set $A$). Then for every $N\geq0$, $\mathbb{P}(\cup_{n\geq N}\{X_n\in\mathcal{O}'\})>0$ and 
\begin{equation*}
\cup_{n\geq N}\{X_n\in\mathcal{O}'\}=\{X_N\in\mathcal{O}'\}\cup\left(\cup_{n>N}\{X_{k}\notin\mathcal{O}',\ \text{for }N\leq k\leq n-1, X_n\in \mathcal{O}'\}\right),
\end{equation*}
where, the union in the R.H.S. is disjoint. Then by Theorem \ref{mnres}, for every $N\geq N_0$,
\begin{align*}
\mathbb{P}(\{X_n\to A\ \text{as }n\to\infty\}&\cap(\cap_{N\geq0}\cup_{n\geq N}\{X_n\in\mathcal{O}'\}))\\&\geq\sum_{n\geq N}\mathbb{P}\left(\{X_n\to A\ \text{as }n\to\infty\}\cap\{X_{k}\notin\mathcal{O}',\ \text{for }N\leq k\leq n-1, X_n\in \mathcal{O}'\}\right)\\
&=\sum_{n\geq N}\bigg[\mathbb{P}\left(\{X_n\to A\ \text{as }n\to\infty|\{X_{k}\notin\mathcal{O}',\ \text{for }N\leq k\leq n-1, X_n\in \mathcal{O}'\}\right)\\
&\ \ \ \ \ \ \ \ \ \ \ \ \mathbb{P}\left(\{X_{k}\notin\mathcal{O}',\ \text{for }N\leq k\leq n-1, X_n\in \mathcal{O}'\}\right)\bigg]\\
&\geq \sum_{n\geq N}\left(1-2de^{-\tilde{K}/b(n)}\right)\mathbb{P}\left(\{X_{k}\notin\mathcal{O}',\ \text{for }N\leq k\leq n-1, X_n\in \mathcal{O}'\}\right)\\
&\geq \left(1-2de^{-\tilde{K}/b(N)}\right)\sum_{n\geq N}\mathbb{P}\left(\{X_{k}\notin\mathcal{O}',\ \text{for }N\leq k\leq n-1, X_n\in \mathcal{O}'\}\right)\\
&=\left(1-2de^{-\tilde{K}/b(N)}\right)\mathbb{P}\left(\cup_{n\geq N}\{X_n\in\mathcal{O}'\}\right)\\
&\geq \left(1-2de^{-\tilde{K}/b(N)}\right)\mathbb{P}\left(\cap_{N\geq 0}\cup_{n\geq N}\{X_n\in\mathcal{O}'\}\right).
\end{align*}
The above inequality is true for every $N\geq N_0$. Taking limit and using the fact that $\lim_{n\to\infty}b(n)=0$, we get that,
\begin{equation*}
\mathbb{P}(\{X_n\to A\ \text{as }n\to\infty\}\cap(\cap_{N\geq0}\cup_{n\geq N}\{X_n\in\mathcal{O}'\}))=\mathbb{P}\left(\cap_{N\geq 0}\cup_{n\geq N}\{X_n\in\mathcal{O}'\}\right).
\end{equation*}
Therefore from the above we can conclude that,
\begin{corollary}
\label{alt_conv}
Under assumptions $(A1)-(A4)$, for almost every $\omega\in\cap_{N\geq0}\cup_{n\geq N}\{X_n\in\mathcal{O}'\}$, $X_n(\omega)\to A$ as $n\to\infty$.
\end{corollary}
\begin{remark}
In comparison with Theorem \ref{aptls}$(a)$, the condition that $\omega\in\cap_{N\geq0}\cup_{n\geq N}\{X_n\in\mathcal{O}'\}$ is stronger than the requirement that $\omega\in \{\lambda(\omega)\cap B(A)\neq\emptyset\}$ because the former requires the iterate sequence to enter an open neighborhood of $A$ with compact closure infinitely often while the latter requires the iterates to enter the basin of attraction of $A$ infinitely often which is larger than $\mathcal{O}'$. But in the presence of stability we have that, 
\begin{equation*}
\{\sup_{n\geq0}\|X_n\|<\infty\}\cap\{\lambda(\cdot)\cap B(A)\neq\emptyset\}\subseteq\cap_{N\geq0}\cup_{n\geq N}\{X_n\in\mathcal{O}'\}.
\end{equation*} 
Further, as a consequence of Corollary \ref{alt_conv}, we have that,
\begin{equation*}
 \mathbb{P}(\{\sup_{n\geq0}\|X_n\|<\infty\}\cap\{\lambda(\cdot)\cap B(A)\neq\emptyset\})=\mathbb{P}(\cap_{N\geq0}\cup_{n\geq N}\{X_n\in\mathcal{O}'\}),
 \end{equation*}
 or in other words, the sample paths which visit $\mathcal{O}'$ infinitely often and are unstable, occur with zero probability.
\end{remark}
\end{itemize} 
\section{Application: Stabilization via resetting}
\label{appl}
In this section we modify recursion $\eqref{rec}$ in such a way that the modified procedure yields sample paths which are stable (that is lie in a compact set almost surely) which in turn allows us to recover the convergence result as in Theorem \ref{aptls}$(b)$ without the need to verify stability, in the presence of a globally attracting set for the mean field. That is, we replace assumption $(A4)$ with the following stronger requirement.

\begin{itemize}
\item [(A4)'] Let $A\subseteq\mathbb{R}^d$ be a globally attracting set for the flow of DI $\eqref{mfld}$.
\end{itemize}

The modification that we propose involves resetting the iterates at regular time intervals if they are found to be lying outside a certain compact set. Let the initial condition $X_0(\omega)=x_0\in\mathbb{R}^d$ for every $\omega\in \Omega$ and $\{r_n\in (0,\infty)\}_{n\geq0}$ be such that, 
\begin{itemize}
\item [(1)] $\|x_0\|<r_0$,
\item [(2)] for every $n\geq0$, $r_n<r_{n+1}$,
\item [(3)] $\lim_{n\to\infty}r_n=\infty$.
\end{itemize}
The modified scheme, henceforth referred to as \it{stabilized stochastic recursive inclusion }\rm (SSRI) is where every sample path is generated as outlined in Algorithm \ref{ssri}.

\begin{algorithm}[H]
\caption{SSRI given $x_0$ and $\{r_k\}_{k\geq0}$}
\label{ssri}
\begin{algorithmic}
\State $n\gets0$\Comment{Initialize iteration count}
\State $k\gets0$\Comment{Initialize reset count}
\State $t_e\gets0$\Comment{Initialize time elapsed since last check}
\State $T_W>0$\Comment{Initialize window length}
\State $n_W\gets1$\Comment{Initialize window count}
\State ${X'}_0(\omega)\gets x_0$\Comment{Initialize initial condition}
\While{$n\geq0$}
\State $X_{n+1}(\omega)-X_n'(\omega)-a(n)M_{n+1}(\omega)\in a(n)F(X_n'(\omega))$\Comment{Obtain $X_{n+1}(\omega)$}
\State $t_e\gets t_e+a(n)$\Comment{Update the time elapsed}
\If{$t_e\geq T_W$}\Comment{Is time elapsed greater than the window length?}
\If{$n_W=1$}\Comment{Have sufficient number of windows elapsed?}
\If{$\|X_{n+1}(\omega)\|> r_k$}\Comment{Is the iterate lying outside a compact set?}
\State $X_{n+1}'(\omega)=x_0$\Comment{Perform reset}
\State $k\gets k+1$\Comment{Increment reset count}
\Else
\State $X_{n+1}'(\omega)\gets X_{n+1}(\omega)$\Comment{Perform no reset}
\EndIf
\State $n_W=2^k$\Comment{Reset window count}
\Else
\State $n_W\gets n_W-1$\Comment{Decrement window count}
\State $X_{n+1}'(\omega)=X_{n+1}(\omega)$\Comment{Perform no reset}
\EndIf
\State $t_e\gets0$\Comment{Reset time elapsed}
\Else
\State $X_{n+1}'(\omega)\gets X_{n+1}(\omega)$\Comment{Perform no reset}
\EndIf 
\State $n\gets n+1$\Comment{Increment iteration count}
\EndWhile
\end{algorithmic}
\end{algorithm}

A flowchart depicting the flow of control in Algorithm \ref{ssri} is presented in Figure \ref{flwchrt}. In order to understand the algorithm let us consider the scenario where the $k^{th}$ reset has been performed at iteration index $n_0$. Then the algorithm checks whether the iterate lies in the compact set $r_kU$ ( closed ball of radius $r_k$ centered at the origin) after approximately $2^kT_W$ amount of time has elapsed (for the relation between time and iteration index see section \ref{mtvtn}). Now either a reset occurs or the iterate is left unchanged. 
\begin{itemize}
\item [(a)] If the iterate is left unchanged then the next reset check is performed after $2^kT_W$ amount of time has elapsed.
\item [(b)] If the iterate is reset, then, the next check is performed after $2^{k+1}T_W$ amount of time has elapsed.
\end{itemize}
In fact it would suffice if the time between successive reset checks were set to be greater than a certain threshold which is determined by the minimum time needed by the flow of the mean field (that is DI $\eqref{mfld}$) to reach the attracting set $A$ from any initial condition in a compact neighborhood of it. But in practical scenarios one may not be able to compute such a time and hence may not be able to determine the required threshold. This approach of increasing time duration between successive reset checks with increasing reset count allows us to bypass this problem. The choice of exponentially increasing durations is one of convenience as it simplifies notations involved in proving certain results later.

\begin{figure}[h]
\centering
\resizebox{16cm}{17cm}{
\begin{tikzpicture}[node distance=2.5cm]
\node (start) [startstop] {Start};
\node (initialize) [io,below of=start,align=left] {$n\gets0$ (iteration count)\\$k\gets0$ (reset count)\\$t_e\gets0$ (time since last reset check)\\ $T_W>0$ (window length)\\ $n_w\gets1$ (window count)\\ $X_0(\omega)\gets x_0$ (initial condition)};
\node (while) [decision,below of=initialize,align=center,yshift=-1.4cm, inner sep=-4pt]{While the\\ iteration count is\\ non-negative, i.e., \\ Is $n\geq0$?};
\node (iteration) [process, below of=while,yshift=-1.2cm,align=center]{Obtain $X_{n+1}(\omega)$ such that,\\$X_{n+1}(\omega)-X_n'(\omega)-a(n)M_{n+1}(\omega)\in a(n)F(X_n'(\omega))$\\ and\\ update time elapsed since the last check, i.e.,\\ $t_e\gets t_e+a(n)$.};
\node (window_length) [decision, below of=iteration,align =center,inner sep=-8pt,yshift=-1.7cm]{Is the\\time elapsed since\\the last check greater\\ than the window\\length?, i.e.,\\ Is $t_e\geq T_W$?};
\node (window_count) [decision,right of=window_length,align=center,inner sep=-4.55pt,xshift=2.3cm]{Have sufficient\\number of windows\\elapsed?, i.e.,\\ Is $n_W=1$?};
\node (norm_check) [decision, right of = window_count,align=center,inner sep=-3.2pt,xshift=2.5cm]{Is the iterate\\lying outside the ball\\of radius $r_k$?, i.e.,\\Is $\|X_{n+1}(\omega)\|>r_k$?};
\node (n_reset_1) [process, below of =window_count, align=center,yshift=-1.5cm]{Perform no reset, i.e.,\\$X_{n+1}'(\omega)\gets X_{n+1}(\omega)$\\and decrement\\window count, i.e.,\\$n_W\gets n_W-1$};
\node (n_reset_2) [process, below of=norm_check,align=center,yshift=-1.5cm]{Perform no reset,i.e,\\$X_{n+1}'(\omega)\gets X_{n+1}(\omega)$};
\node (reset) [process, right of=n_reset_2,align=center,xshift=1.5cm]{Perform reset, i.e.,\\$X_{n+1}'(\omega)\gets x_0$\\and increment\\reset count, i.e.,\\$k\gets k+1$};
\node (reset_window_count) [process, below of=n_reset_2, align=center,yshift=-1.5cm]{Reset window\\count, i.e.,\\$n_W\gets2^k$};
\node (reset_time_elapsed) [process, below of=n_reset_1,align=center,yshift=-1.5cm]{Reset time\\elapsed, i.e.,\\$t_e\gets0$};
\node (n_reset_3) [process, below of=window_length,align=center,yshift=-1.5cm]{Perform no reset,i.e,\\$X_{n+1}'(\omega)\gets X_{n+1}(\omega)$};
\node (iteration_count_inc) [process, below of=n_reset_3,align=center,yshift=-1.5cm]{Increment iteration\\count, i.e.,\\$n\gets n+1$};
\node (stop) [startstop,right of =while,xshift=2cm]{Stop};
\draw [arrow] (start)--(initialize);
\draw [arrow] (initialize)--(while);
\draw [arrow] (while)--node[anchor=east]{yes}(iteration);
\draw [arrow] (iteration)--(window_length);
\draw [arrow] (window_length)--node[anchor=east]{no}(n_reset_3);
\draw [arrow] (n_reset_3)--(iteration_count_inc);
\draw [arrow] (window_count)--node[anchor=east]{no}(n_reset_1);
\draw [arrow] (norm_check)--node[anchor=east]{no}(n_reset_2);
\draw [arrow] (n_reset_1)--(reset_time_elapsed);
\draw [arrow] (n_reset_2)--(reset_window_count);
\draw [arrow] (while)--node[anchor=south]{no}(stop);
\draw [arrow] (window_length)--node[anchor=south]{yes}(window_count);
\draw [arrow] (window_count)--node[anchor=south]{yes}(norm_check);
\draw [arrow] (norm_check)-|node[anchor=south]{yes}(reset);
\draw [arrow] (reset)|-(reset_window_count);
\draw [arrow] (reset_window_count)--(reset_time_elapsed);
\draw [arrow] (reset_time_elapsed)--(iteration_count_inc);
\coordinate (left while) at ($(while.west)+(-2.75cm,0)$);
\coordinate (left iter_count_inc) at ($(iteration_count_inc.west) + (-3.0cm,0)$);
\draw [arrow] (iteration_count_inc.west)--(left iter_count_inc)--(left while)--(while.west);
\end{tikzpicture}
}
\caption{Flowchart depicting the flow of control in Algorithm \ref{ssri}}\label{flwchrt}
\end{figure}
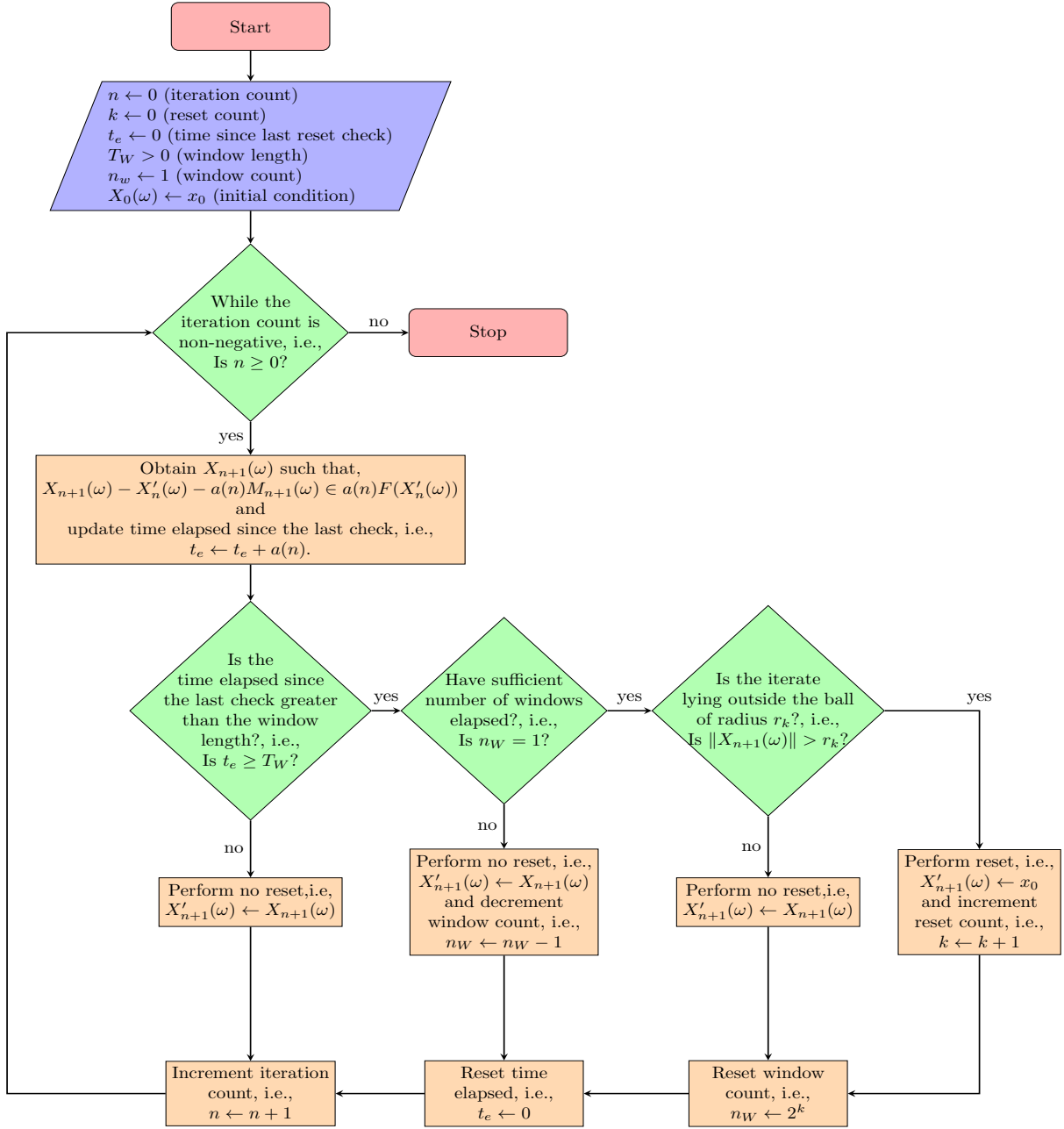

For every $n\geq1$, define the indicator random variable $\chi_n:\Omega\rightarrow\{0,1\}$ such that, for every $\omega\in\Omega$, 
\begin{equation}
\label{indi}
\chi_n(\omega)=\begin{cases} 
                            0 &\mathrm{if }\ X_n(\omega)=X_n'(\omega), \\
                            1 &\mathrm{if }\ X_n(\omega)\neq X_n'(\omega). 
                          \end{cases}
\end{equation}
We assume that the noise terms $\{M_n\}_{n\geq1}$ satisfy the following version of assumption $(A3)$. 
\begin{itemize}
\item [(A3)'] $\{M_n\}$ is a martingale difference sequence with respect to the filtration $\{\mathscr{F}_n\}_{n\geq1}$, where, for every $n\geq1$, $\mathscr{F}_n$ denotes the smallest $\sigma$-algebra generated by the iterates $X_m$ (that is the iterates before the reset operation) and noise terms $M_m$, for $0\leq m\leq n$ (then it is easy to show that for every $n\geq1$, $X_n'$  and hence $\chi_n$ are $\mathscr{F}_n$ measurable). Since for every $n\geq1$, $M_n$ denotes the noise arising in the estimation (or measurement) of $F$ at $X_{n-1}'$, we assume that the energy of the noise depends on $X_{n-1}'$. That is for every $n\geq0$, $\|M_{n+1}\|\leq K(1+\|X_n'\|)\ a.s.$. 
\end{itemize}

The next theorem says that, for almost every sample path generated by Algorithm \ref{ssri}, the total number of resets is finite, thereby guaranteeing stability. The proof of this theorem (provided in section \ref{prf_fnrst}) crucially hinges on a lower bound for the probability of the event that there are no future resets given that there are a certain number of resets up until iteration $n_0$ for some large $n_0$. Specifically it requires the probability of the above mentioned event to converge to one as $n_0$ tends to infinity and this is guaranteed by Theorem \ref{mnres}.
\begin{theorem}
\label{fnrst}
\emph{(Finite resets)}
Under assumptions $(A1), (A2), (A3)'$ and $(A4)'$, $\mathbb{P}\left(\left\{\omega\in\Omega: \sum_{n=1}^{\infty}\chi_n(\omega)<\infty\right\}\right)\\=1$.
\end{theorem}

As a consequence of the above theorem, we have the following.
\begin{itemize}   
\item [(a)] Let $\omega\in\{\omega\in\Omega: \sum_{n=1}^{\infty}\chi_n(\omega)<\infty\}$. Then there exists an $N\geq1$ and $R>0$ (depending on $\omega$) such that, for every $n\geq N$, $X_n(\omega)=X_n'(\omega)$ and $\sup_{n\geq N}\|X_n(\omega)\|\leq R$. Therefore $\sum_{n\geq N}\mathbb{E}[(a(n))^2\|M_{n+1}\|^2|\mathscr{F}_n](\omega)\leq \sum_{n\geq N}(a(n))^2K^2\left(1+\|X_n'(\omega)\|\right)^2\leq K^2(1+R)^2\sum_{n\geq N}(a(n))^2\\<\infty$, where the last inequality follows from assumption $(A2)$(ii). Therefore,
\begin{equation*}
\{\omega\in\Omega:\sum_{n\geq1}\chi_n(\omega)<\infty\}\subseteq\{\omega\in\Omega:\sum_{n=0}^{\infty}\mathbb{E}[(a(n))^2\|M_{n+1}\|^2|\mathscr{F}_n](\omega)<\infty\}.
\end{equation*}
Therefore by Theorem \ref{fnrst}, we have that the $\mathbb{P}(\sum_{n=0}^{\infty}\mathbb{E}[(a(n))^2\|M_{n+1}\|^2|\mathscr{F}_n]<\infty)=1$ and by martingale convergence theorem (see \cite[Section~11.3, Thm.~11]{borkartxt}) we have that, the square integrable martingale $\{\sum_{m=0}^{n-1}a(m)M_{m+1},\mathscr{F}_n\}_{n\geq1}$ converges almost surely. 
\item [(b)] Thus for $\omega$ lying in a probability one set, there exists $N\geq1$ and $R>0$ (depending on $\omega$) such that along this sample path the iterates $\{X_n(\omega)\}_{n\geq N}$, can be viewed as being generated by recursion $\eqref{rec}$ with initial condition $X_N(\omega)$, their norms are bounded by $R$ uniformly and the additive noise terms $\{M_n(\omega)\}_{n\geq N}$ satisfy the hypothesis of \cite[Prop.~1.3]{benaim1}. Then by arguments similar to those of Theorem \ref{aptls}$(b)$ we have that,
\begin{corollary}
Under assumptions $(A1)-(A3)$ and $(A4)'$, for almost every $\omega$, the iterates generated by Algorithm \ref{ssri}, $\{X_n'(\omega)\}_{n\geq0}$, are such that $X_n'(\omega)\to A$ as $n\to\infty$.
\end{corollary}
\end{itemize}

\section{Proof of the lock-in probability theorem (Thm.~\ref{mnres})}
\label{prf_lip}
Proof of the lock-in probability result follows as a consequence of a series of lemmas. The overall structure can be summarized as follows.
\begin{itemize}
\item [(a)] Our first aim is to replace the set-valued map in recursion $\eqref{rec}$ with an equivalent single-valued locally Lipschitz continuous function with an additional parameter. In order to accomplish this, we first embed the graph of the set-valued map $F$ in the graph of a sequence of locally Lipschitz continuous set-valued maps. These maps are then parametrized using the Stiener selection procedure which preserves the modulus of continuity.
\item [(b)] The relation between the solutions of DI $\eqref{mfld}$ and that of differential inclusions with continuous set-valued maps which approximate $F$ (as in (a) above) is established. \item [(c)] An ordinary differential equation(o.d.e) is defined using an appropriate single-valued parametrization of $F$ (as in (a) above). The existence of solutions to such an o.d.e. and further its uniqueness follow from Caratheodary's existence theorem and locally Lipschitz nature of the vector field respectively. The solutions of this o.d.e. aide in separating the probability contributions due to the additive noise terms and the set-valued nature of the drift function. Using the results from part (b) above, we conclude that after a large number of iterations, the probability contribution is only due to the additive noise terms.
\item [(d)] We finally review the standard probability lower bounding procedure for the additive noise terms from \cite[Ch.~4]{borkartxt}. Using this bound in the result obtained in part (c) 
above gives us the desired lock-in probability bound.
\end{itemize}
Throughout, we use $U$ to denote the closed unit ball in $\mathbb{R}^d$ centered at the origin. Further, for every $Y_1,Y_2\subseteq \mathbb{R}^d$ and $r\in\mathbb{R}$, define,
\begin{itemize}
\item $Y_1+Y_2:=\{y_1+y_2:y_1\in Y_1\ \mathrm{and}\ y_2\in Y_2\}$,
\item $rY_1:=\{ry_1:y_1\in Y_1\}$.
\end{itemize}
%--------------------------------------------------------------------------------------------------------------------------------------------------------------------------------------------------------------------------------------

\subsection{Upper semicontinuous set-valued maps and their approximation}
\label{usc_app}
First we recall definitions of continuous set-valued maps and locally Lipschitz continuous set-valued maps. These notions are taken from \cite[Ch.~1]{aubindi}.
\begin{definition}
A set-valued map $F:\mathbb{R}^d\rightarrow\{\text{compact subsets of }\mathbb{R}^d\}$ is,
\begin{itemize}
\item \it{upper semicontinuous }\rm(u.s.c.) if, for every $x\in\mathbb{R}^d$, for every $\epsilon>0$, there exists a $\delta>0$ (depending on $x$ and $\epsilon$) such that, for every $x'\in\mathbb{R}^d$ satisfying $\|x'-x\|<\delta$, we have that $F(x')\subseteq F(x)+\epsilon U$, where $F(x)+\epsilon U:=\{y+\epsilon u: y\in F(x),\ u\in U\}$.
\item \it{lower semicontinuous }\rm(l.s.c.) if, for every $x\in\mathbb{R}^d$, for every $\mathbb{R}^d$-valued sequence $\{x_n\}_{n\geq1}$ converging to $x$, for every $y\in F(x)$, there exists a sequence $\{y_n\in F(x_n)\}_{n\geq1}$ converging to $y$.
\item \it{continuous }\rm if, it is both u.s.c. and l.s.c.
\item \it{locally Lipschitz continuous }\rm if, for every $x_0\in\mathbb{R}^d$, there exists $\delta>0$ and $L>0$ (depending on $x_0$) such that for every $x,x'\in x_0+\delta U$, we have that 
$F(x)\subseteq F(x')+L\|x-x'\|U$.
\end{itemize} 
\end{definition} 

Let $\mathcal{K}(\mathbb{R}^d)$ denote the family of all non-empty compact subsets of $\mathbb{R}^d$. Let $\bm{\mathrm{H}}:\mathcal{K}(\mathbb{R}^d)\times\mathcal{K}(\mathbb{R}^d)\rightarrow[0,\infty)$ be defined such that, for every $S_1,\ S_2\in \mathcal{K}(\mathbb{R}^d)$, 
\begin{equation}
\bm{\mathrm{H}}(S_1,S_2):=\max\left\{\sup_{s_1\in S_1}\inf_{s_2\in S_2}\left\|s_1-s_2\right\|,\ \sup_{s_2\in S_2}\inf_{s_1\in S_1}\left\|s_1-s_2\right\|\right\}.
\end{equation}

With $\bm{\mathrm{H}}$ as defined above, $(\mathcal{K}(\mathbb{R}^d),\bm{\mathrm{H}})$ is a complete metric space (for a proof see \cite[Thm.~1.1.2]{shoumei}). The notions of continuity and local Lipschitz continuity of a set-valued map can be restated using the metric defined above and is stated as a lemma below for easy reference (for a proof see \cite[Ch.~1, section~5, Cor.~1]{aubindi}).
\begin{lemma}
A set-valued map $F:\mathbb{R}^d\rightarrow\mathcal{K}(\mathbb{R}^d)$ is
\begin{itemize}
\item [(a)] Continuous, if and only if, for every $x_0\in\mathbb{R}^d$, for every $\epsilon>0$, there exists $\delta>0$ (depending on $x_0$ and $\epsilon$), such that for every $x\in x_0+\delta U$, $\bm{\mathrm{H}}(F(x),F(x_0))<\epsilon$.
\item [(b)] locally Lipschitz continuous, if and only if, for every $x_0\in\mathbb{R}^d$, there exists $\delta>0$ and $L>0$ (depending on $x_0$), such that for every $x,\ x'\in x_0+\delta U$, 
$\bm{\mathrm{H}}(F(x),F(x'))\leq L\|x-x'\|$.
\end{itemize} 
\end{lemma}

Before we proceed further we look at a certain form of locally Lipschitz continuous set-valued maps that arise later. The next lemma defines such maps and also states that the sum of two locally Lipschitz continuous set-valued maps is again a locally Lipschitz continuous set-valued map, a result needed later to obtain locally Lipschitz continuous single-valued parametrization of map $F$ in recursion $\eqref{rec}$.

\begin{lemma}
\label{llc}
\begin{itemize}
\item [(a)] If $f:\mathbb{R}^d\rightarrow\mathbb{R}$ is a locally Lipschitz continuous map and $C\in\mathcal{K}(\mathbb{R}^d)$, then the set-valued map $F:\mathbb{R}^d\rightarrow\mathcal{K}(\mathbb{R}^d)$, given by $F(x):=f(x)C$ for every $x\in\mathbb{R}^d$, is a locally Lipschitz continuous set-valued map. 
\item [(b)] If for every $i\in\{1,2\}$, $F_i:\mathbb{R}^d\rightarrow\mathcal{K}(\mathbb{R}^d)$ is a locally Lipschitz continuous set-valued map, then the set-valued map $F:\mathbb{R}^d\rightarrow\mathcal{K}(\mathbb{R}^d)$, given by $F(x):=F_1(x)+F_2(x)$ for every $x\in\mathbb{R}^d$, is a locally Lipschitz continuous set-valued map.
\end{itemize}
\end{lemma}
\begin{proof}
\begin{itemize}
\item [(a)] Fix $x_0\in\mathbb{R}^d$ and let $r:=\sup_{c\in C}\|c\|$. Since $f$ is locally Lipschitz continuous, there exists $\delta^f_{x_0}>0$ and $L^f_{x_0}>0$ such that for every $x,x'\in x_0+\delta^f_{x_0}U$, $|f(x)-f(x')|\leq L^f_{x_0}\|x-x'\|$. Let $x,x'\in x_0+\delta^{f}_{x_0}U$. Then for any $c\in C$, 
\begin{align*}
\|f(x)c-f(x')c\|&=|f(x)-f(x')|\|c\|\\
                      &\leq rL^f_{x_0}\|x-x'\|.
\end{align*}
Therefore for every $x,x'\in x_0+\delta^{f}_{x_0}U$, for every $c\in C$, $f(x')c-f(x)c\in rL^{f}_{x_0}\|x-x'\|U$. Thus for every $x,x'\in x_0+\delta^{f}_{x_0}U$, $F(x')\subseteq F(x)+rL^{f}_{x_0}\|x-x'\|U$, from which it follows that the set-valued map $F$ is locally Lipschitz continuous at $x_0$ with $\delta:=\delta^f_{x_0}$ and $L:=rL^{f}_{x_0}$. Since $x_0\in\mathbb{R}^d$ is arbitrary, the above argument gives us that $F$ is locally Lipschitz continuous at every $x_0$.

\item [(b)] Fix $x_0\in\mathbb{R}^d$. Since for every $i\in\{1,2\}$, $F_i$ are locally Lipschitz continuous, there exists $\delta_i>0$ and $L_i>0$ such that for every $x,x'\in x_0+\delta_i U$, $F_i(x)\subseteq F_i(x')+L_i\|x-x'\| U$. Let $\delta:=\min\{\delta_1,\delta_2\}$, $L:=L_1+L_2$ and $x,x'\in x_0+\delta U$. For any $y\in F(x)$, there exists $y_1\in F_1(x)$ and $y_2\in F_2(x)$ such that $y=y_1+y_2$ . By our choice of $\delta$, we have $y'_1\in F_1(x')$, $y'_2\in F_2(x')$ and $u_1, u_2\in U$, such that for every $i\in\{1,2\}$, $y_i=y'_i+L_i\|x-x'\|u_i$. Therefore,
\begin{align}
\label{tmp0}
y&=y_1+y_2\nonumber\\\nonumber
  &=y'_1+L_1\|x-x'\|u_1 +y'_2+L_2\|x-x'\|u_2\\
  &=y'_1+y'_2+\left(L_1+L_2\right)\left\|x-x'\right\|\left(\frac{L_1u_1+L_2u_2}{L_1+L_2}\right).
\end{align}
Clearly $y'_1+y'_2\in F(x')$ and since $U$ is a convex subset of $\mathbb{R}^d$, $\frac{L_1u_1+L_2u_2}{L_1+L_2}\in U$. From $\eqref{tmp0}$ we get that, $F(x)\subseteq F(x')+(L_1+L_2)\|x-x'\| U$, for every $x,x' \in x_0+\delta U$. Therefore $F$ is locally Lipschitz continuous at $x_0$. Since $x_0$ is arbitrary, the above argument gives us that $F$ is locally Lipschitz continuous.\qed  
\end{itemize}
\end{proof}  

Consider a set-valued map $F$ satisfying assumption $(A1)$. A simple contradiction argument gives us that $F$ is u.s.c. It is not possible to represent such u.s.c. set-valued maps with a single-valued continuous map with an additional parameter. But instead one can approximate them from above as explained next. The first step is to embed the graph of the map $F$ in that of a sequence of continuous set-valued maps as stated in the lemma below. For the proof of the lemma below notions of a paracompact topological space, an open covering, its locally finite refinement and partition of unity subordinated to a locally finite covering are needed, which are summarized in Appendix \ref{topo} for easy reference.

\begin{lemma}
\label{ctem}
Let $F:\mathbb{R}^d\rightarrow\mathcal{K}(\mathbb{R}^d)$ be a set-valued map satisfying $(A1)$. Then, there exists a sequence of continuous set-valued maps $\{F^{(l)}:\mathbb{R}^d\rightarrow\mathcal{K}(\mathbb{R}^d)\}_{l\geq1}$, such that for every $l\geq1$,
\begin{itemize}
\item [(a)] for every $x\in\mathbb{R}^d$, $F^{(l)}(x)$ is a non-empty, convex and compact subset of $\mathbb{R}^d$,
\item [(b)] for every $x\in\mathbb{R}^d$, $F(x)\subseteq F^{(l+1)}(x)\subseteq F^{(l)}(x)$,
\item [(c)] there exists $K^{(l)}>0$, such that for every $x\in\mathbb{R}^d$, $\sup_{y\in F^{(l)}(x)}\|y\|\leq K^{(l)}(1+\|x\|)$,
\item [(d)] $F^{(l)}$ is a locally Lipschitz continuous set valued map.
\end{itemize}
Furthermore, 
\begin{itemize}
\item [(e)] for every $x\in\mathbb{R}^d$, $F(x)=\cap_{l\geq1} F^{(l)}(x)$.
\end{itemize}
\end{lemma} 
\begin{proof}
For any $\epsilon>0$, for every $x_0\in\mathbb{R}^d$, let $B(\epsilon,x_0):=\left\{x:\ \parallel x-x_0\parallel<\epsilon\right\}$. Let $\left\{\epsilon_l:=\frac{1}{3^l}\right\}_{l\geq1}$. Then for every $l\geq1$, $\mathscr{C}_l:=\left\{B(\epsilon_l,x_0):x_0\in\mathbb{R}^d\right\}$ is an open covering of $\mathbb{R}^d$. Since $\mathbb{R}^d$ is a metric space, it is paracompact (see \cite[Ch.~0, Sec.~1, Thm.~1]{aubindi}). Therefore for every $l\geq1$, there exists a locally finite open refinement of the covering $\mathscr{C}_l$ and let it be denoted by $\tilde{\mathscr{C}}_l:=\left\{C_i^l\right\}_{i\in I^{l}}$ where $I^l$ is an arbitrary index set. By \cite[Ch.~0, Sec.~1, Thm.~2]{aubindi}, there exists a locally Lipschitz continuous partition of unity, $\left\{\psi_i^{l}\right\}_{i\in I^l}$, subordinated to the covering $\tilde{\mathscr{C}}_l$. Therefore, 
for every $l\geq1$, for every $i\in I^l$, there exists $x_i^l$, such that $\mathrm{support}(\psi_i^l)\subseteq C_i^l\subseteq B(\epsilon_l,x_i^l)$. For every $l\geq1$, for every $x\in\mathbb{R}^d$, let $I^l(x):=\left\{i\in I^l:\psi_i^l(x)>0\right\}$ and by definition of $\psi_i^l$, we have that $0<|I^l(x)|<\infty$ and $\sum_{i\in I^l(x)}\psi_i^l(x)=1$.

For every $l\geq1$, define the set valued map $F^{(l)}:\mathbb{R}^d\rightarrow\left\{\text{subsets of }\mathbb{R}^d\right\}$, 
such that for every $x\in\mathbb{R}^d$, $F^{(l)}(x):=\sum_{i\in I^l(x)}\psi_i^{l}(x)A_i^{l}$, where $A_i^{l}:=\bar{co}\left(F\left(B\left(2\epsilon_l,x_i^l\right)\right)\right)$. 

The proofs of parts $(a),\ (b),\ (c)$ and $(e)$ of the lemma are exactly the same as that of \cite[Lemma~3.2]{vin1t}. We shall provide a proof of part $(d)$ of the lemma above from which continuity of the set-valued maps $F^{(l)}$ follows.
\begin{itemize}
\item [(d)] Fix $l\geq1$ and $x\in\mathbb{R}^d$. Since $\tilde{\mathscr{C}}_l$ is a locally finite open covering of $\mathbb{R}^d$, there exists $\delta>0$ (depending on $x$), such that 
$I^l(x,\delta):=\left\{i\in I^l:\ B(x,\delta)\cap C^l_i\neq\emptyset\right\}$ is finite. Since $\{\psi_i^l\}_{i\in I^l}$ is a locally Lipschitz continuous partition of unity subordinated to the covering $\tilde{\mathscr{C}}_l$, we have that for every $i\in I^l$, $\mathrm{support}(\psi_i^l)\subseteq C_i^l$. Therefore, for every $x'\in B(x,\delta)$, $F^{(l)}(x')=\sum_{i\in I^l(x,\delta)}\psi^l_i(x')A_i^l$.

From the proof of part $(a)$ of this lemma we know that for every $i\in I^l$, $A_i^l$ is a compact and convex subset of $\mathbb{R}^d$. Therefore from Lemma \ref{llc}$(a)$, we get that , for every $i\in I^l(x,\delta)$, the set-valued map given by $y\rightarrow \psi_i^l(y)A_i^l$ is locally Lipschitz continuous. Further since $|I^l(x,\delta)|<\infty$, from Lemma \ref{llc}$(b)$, we get that the set-valued map given by $y\rightarrow \sum_{i\in I^l(x,\delta)}\psi^l(y)A_i^l$ is locally Lipschitz continuous. Since the set-valued map $y\rightarrow \sum_{i\in I^l(x,\delta)}\psi^l_i(y)A_i^l$ restricted to $B(x,\delta)$ is the same as $F^{(l)}$ on $B(x,\delta)$, we get that $F^{(l)}$ is locally Lipschitz continuous at $x$. Since $x$ is arbitrary, the above argument gives us that $F^{(l)}$ is a locally Lipschitz continuous set-valued map.\qed 
\end{itemize}
\end{proof}

The continuous set-valued maps $F^{(l)}$ as obtained above can be now parametrized (that is represented with a single-valued continuous function with an additional parameter). Key to parametrization is a continuous selection procedure by which we mean a function $\sigma:\mathcal{K}(\mathbb{R}^d)\rightarrow\mathbb{R}^d$ which is continuous and is such that for every $Y\in\mathcal{K}(\mathbb{R}^d)$, $\sigma(Y)\in Y$. Since the maps $F^{(l)}$ are convex set-valued, it suffices to look for a selection procedure which is continuous restricted to the family of compact and convex subsets of $\mathbb{R}^d$. Further we want a selection procedure which would preserve the local Lipschitz continuity of the set-valued map $F^{(l)}$ in the parametrization as well. In order to accomplish this we shall use the Stiener selection procedure (for a definition see \cite[Thm.~9.4.1]{aubinsva}). The next lemma summarizes some properties of the Stiener selection procedure and an intersection lemma which form the central tools for parameterizing the set-valued maps $F^{(l)}$ (for a proof we refer the reader to \cite[Thm.~9.4.1]{aubinsva} and \cite[Lemma~9.4.2]{aubinsva}). Before we state the lemma we introduce some notation needed. Let $\mathcal{K}_c(\mathbb{R}^d)$ denote the family of all non-empty compact and convex subsets of $\mathbb{R}^d$. For any set $Y\subseteq\mathbb{R}^d$ and for any $x\in\mathbb{R}^d$, define $\bm{d}(x,Y):=\inf_{y\in Y}\|x-y\|$.

\begin{lemma}
\label{prmtl}
\begin{itemize}
\item [(a)] There exists a function $\sigma:\mathcal{K}_c(\mathbb{R}^d)\rightarrow\mathbb{R}^d$, such that for every $Y,Y_1, Y_2\in\mathcal{K}_c(\mathbb{R}^d)$, 
\begin{equation*}
\sigma(Y)\in Y\ \mathrm{and}\ \|\sigma(Y_1)-\sigma(Y_2)\|\leq d\ \bm{\mathrm{H}}(Y_1,Y_2).
\end{equation*}
\item [(b)] The map $\Pi:\mathcal{K}_c(\mathbb{R}^d)\times \mathbb{R}^d\rightarrow \mathcal{K}_c(\mathbb{R}^d)$, defined such that for every $Y\in\mathcal{K}_c(\mathbb{R}^d)$ and $x\in\mathbb{R}^d$, $\Pi(Y,u):=Y\cap\left(x+2\bm{d}(x,Y)U\right)$, is such that for every $Y_1,Y_2\in\mathcal{K}_c(\mathbb{R}^d)$ and for every $x_1,x_2\in\mathbb{R}^d$, 
\begin{equation*}
\bm{\mathrm{H}}(\Pi(Y_1,x_1),\Pi(Y_2,x_2))\leq5\left(\bm{\mathrm{H}}(Y_1,Y_2)+\|x_1-x_2\|\right).
\end{equation*}
\end{itemize}
\end{lemma}

We now use the results stated in the above lemma to parametrize the set-valued maps $F^{(l)}$.

\begin{lemma}
\label{param}
Let $\{F^{(l)}\}_{l\geq1}$ be as in Lemma \ref{ctem}. For every $l\geq1$, there exists a continuous function $f^{(l)}:\mathbb{R}^d\times U\rightarrow \mathbb{R}^d$ such that, 
\begin{itemize}
\item [(a)] for every $x\in\mathbb{R}^d$, $f^{(l)}(x,U)=F^{(l)}(x)$ where $f^{(l)}(x,U):=\left\{f^{(l)}(x,u):\ u\in U\right\}$.
\item [(b)] for $K^{(l)}>0$ as in Lemma \ref{ctem}, for every $(x,u)\in\mathbb{R}^d\times U$, $\|f^{(l)}(x,u)\|\leq K^{(l)}(1+\|x\|)$.
\item [(c)] for every $x_0\in\mathbb{R}^d$, there exists $\delta^{(l)}>0$ and $L^{(l)}>0$ (depending on $x_0$), such that for every $x,x'\in x_0+\delta^{(l)} U$, for every $u\in U$, 
\begin{equation*}
\left\|f^{(l)}(x,u)-f^{(l)}(x',u)\right\|\leq L^{(l)}\|x-x'\|.
\end{equation*}
\end{itemize} 
\end{lemma}
\begin{proof}
Fix $l\geq1$. Let the map $f^{(l)}:\mathbb{R}^d\times U\rightarrow\mathbb{R}^d$ be defined such that, for every $(x,u)\in\mathbb{R}^d\times U$, 
\begin{equation}
\label{prmdef}
f^{(l)}(x,u):=\sigma\left(\Pi\left(F^{(l)}(x),K^{(l)}(1+\|x\|)u\right)\right),
\end{equation}
where $\sigma$ and $\Pi$ are as in Lemma \ref{prmtl}.
\begin{itemize}
\item [(a)] By definition of $f^{(l)}$, $\sigma$ and $\Pi$, for every $(x,u)\in \mathbb{R}^d\times U$, we have that,
\begin{equation*}
f^{(l)}(x,u)\in \Pi(F^{(l)}(x),K^{(l)}(1+\|x\|)u)\subseteq F^{(l)}(x).
\end{equation*} 
Therefore, for every $x\in\mathbb{R}^d$, $f^{(l)}(x,U)\subseteq F^{(l)}(x)$. By Lemma \ref{ctem}$(c)$, we know that for every $x\in\mathbb{R}^d$, $\sup_{y\in F^{(l)}(x)}\|y\|\leq K^{(l)}(1+\|x\|)$. Thus for every $x\in\mathbb{R}^d$,  for any $y\in F^{(l)}(x)$, there exists $u\in U$, such that $y=K^{(l)}(1+\|x\|)u$. For such a $u\in U$, by definition of $\Pi$, we have that $\Pi(F^{(l)}(x),K^{(l)}(1+\|x\|)u)=y$ and hence $f^{(l)}(x,u)=\sigma\left(\Pi\left(F^{(l)}(x),K^{(l)}(1+\|x\|)u\right)\right)=y$. Therefore for every $x\in\mathbb{R}^d$, $F^{(l)}(x)\subseteq f^{(l)}(x,U)$ from which it follows that $f^{(l)}(x,U)=F^{(l)}(x)$, for every $x\in\mathbb{R}^d$.
\item [(b)] Follows from part $(a)$ of this lemma and Lemma \ref{ctem}$(c)$.
\item [(c)] Fix $x_0\in\mathbb{R}^d$. Since $F^{(l)}$ is a locally Lipschitz continuous set-valued map (see Lemma \ref{ctem}$(d)$), we obtain $\delta_{F^{(l)}}>0$ and $L_{F^{(l)}}>0$ (depending on $x_0$) such that for every $x,x'\in x_0+\delta_{F^{(l)}}U$, $\bm{\mathrm{H}}(F^{(l)}(x),F^{(l)}(x'))\leq L_{F^{(l)}}\|x-x'\|$. Set $\delta^{(l)}:=\delta_{F^{(l)}}$ and $L^{(l)}:=5d(L_{F^{(l)}}+K^{(l)})$. Then, for any $x,x'\in x_0+\delta^{(l)}U$, for every $u\in U$,
\begin{align}
\left\|f^{(l)}(x,u)-f^{(l)}(x',u)\right\|&=\left\|\sigma\left(\Pi\left(F^{(l)}(x),K^{(l)}(1+\|x\|)u\right)\right)-\sigma\left(\Pi\left(F^{(l)}(x'),K^{(l)}(1+\|x'\|)u\right)\right)\right\|\nonumber\\
                                                            \label{tmp1} 
                                                            &\leq d\ \bm{\mathrm{H}}\left(\Pi\left(F^{(l)}(x),K^{(l)}(1+\|x\|)u\right),\Pi\left(F^{(l)}(x'),K^{(l)}(1+\|x'\|)u\right)\right)\\
                                                            \label{tmp2}
                                                            &\leq 5d\ \left (\bm{\mathrm{H}}\left( F^{(l)}(x),F^{(l)}(x') \right) +\left \|K^{(l)}\left( 1+\|x\| \right)u-K^{(l)}\left( 1+\|x'\| \right)u \right\| \right)\\
                                                            &= 5d\ \left( \bm{\mathrm{H}}\left( F^{(l)}(x),F^{(l)}(x') \right)+K^{(l)}\left|\|x\|-\|x'\|\right|\|u\| \right)\nonumber\\
                                                            &\leq 5d\ \left( \bm{\mathrm{H}}\left( F^{(l)}(x),F^{(l)}(x') \right)+K^{(l)}\|x-x'\|\right)\nonumber\\
                                                            \label{tmp3}
                                                            &\leq 5d\ \left(L_{F^{(l)}}\|x-x'\|+K^{(l)}\|x-x'\|\right)\\
                                                            &= L^{(l)} \|x-x'\|\nonumber,
\end{align}
where, $\eqref{tmp1}$ follows from Lemma \ref{prmtl}$(a)$, $\eqref{tmp2}$ follows from Lemma \ref{prmtl}$(b)$ and $\eqref{tmp3}$ follows from our choice of $\delta^{(l)}$ and local Lipschitz continuity of $F^{(l)}$.\qed
\end{itemize}
\end{proof}

The set-valued map in recursion $\eqref{rec}$ can be replaced with the parametrization obtained in the lemma above as explained below.
\begin{itemize}
\item [(1)] For every $l\geq 1$, by Lemma \ref{ctem}$(b)$, we know that for every $x\in\mathbb{R}^d$, $F(x)\subseteq F^{(l)}(x)$. Therefore for every $l\geq1$, for every $n\geq0$, 
\begin{equation*}
X_{n+1}-X_{n}-a(n)M_{n+1}\in a(n)F^{(l)}(X_n).
\end{equation*}
\item [(2)] For every $l\geq1$, by Lemma \ref{param}$(a)$, we know that for every $x\in\mathbb{R}^d$, $F^{(l)}(x)=f^{(l)}(x,U)$. It can now be shown that for every $n\geq0$, there exists a $U$-valued random variable on $\Omega$, say $U^{(l)}_n$, such that for every $\omega\in\Omega$, for every $n\geq0$, 
\begin{equation}
\label{rec1}
X_{n+1}(\omega)-X_{n}(\omega)-a(n)M_{n+1}(\omega)= a(n)f^{(l)}(X_n(\omega),U^{(l)}_n(\omega))
\end{equation}
(for a proof see \cite[Lemma~6.1]{vin1t}).
\end{itemize}
%----------------------------------------------------------------------------------------------------------------------------------------------------------------------------------------------------------------------------------

\subsection{Solutions of the mean field and their approximation}
\label{sol_apprx}
In this section, we shall approximate the solutions of mean field (that is DI $\eqref{mfld}$) with the solutions of DI given by,
\begin{equation}
\label{amfld}
\frac{dx}{dt}\in F^{(l)}(x),
\end{equation}
for some $l\geq1$. In order to accomplish this we need some notations which are introduced next. 

For every $T>0$ and for every $x\in\mathbb{R}^d$,  let $S(T,x)$ denote the set of solutions of DI $\eqref{mfld}$ on $[0,T]$. Formally, 
\begin{small}
\begin{equation}
\label{ss0}
S(T,x):=\left\{\bm{\mathrm{x}}:[0,T]\rightarrow\mathbb{R}^d\ :\ \bm{\mathrm{x}}\ \text{is absolutely continuous with $\bm{\mathrm{x}}(0)=x$ and for }a.e.\ t\in[0,T], \frac{d\bm{\mathrm{x}}(t)}{dt}\in F(\bm{\mathrm{x}}(t))\right\}.
\end{equation}
\end{small}
Since $F$ is a Marchaud map, we have that for every $T>0$ and for every $x\in\mathbb{R}^d$, $S(T,x)\neq\emptyset$. Similarly for every $l\geq1$, for every $T>0$ and for every $x\in\mathbb{R}^d$, let $S^{(l)}(T,x)$ denote the set of solutions of DI $\eqref{amfld}$ on $[0,T]$. Formally,
\begin{small}
\begin{equation}
\label{ss1}
S^{(l)}(T,x)\!:=\!\left\{\!\bm{\mathrm{x}}:[0,T]\rightarrow\mathbb{R}^d\ \!\!:\ \!\! \bm{\mathrm{x}}\ \text{is absolutely continuous with $\bm{\mathrm{x}}(0)=x$ and for }a.e.\ t\in[0,T], \frac{d\bm{\mathrm{x}}(t)}{dt}\!\in\! F^{(l)}(\bm{\mathrm{x}}(t))\!\right\}.
\end{equation}
\end{small}
From Lemma \ref{ctem}, we know that for every $l\geq1$, $F^{(l)}$ is a Marchaud map and hence for every $T>0$ and for every $x\in\mathbb{R}^d$, $S^{(l)}(T,x)\neq\emptyset$. 

For any $Y\subseteq\mathbb{R}^d$, for any $T>0$, define $S(T,Y):=\cup_{y\in Y}S(T,y)$. Similarly, for every $l\geq1$, $S^{(l)}(T,Y):=\cup_{y\in Y}S^{(l)}(T,y)$. 

The next lemma summarizes some important relationships between the solutions of DI $\eqref{mfld}$ and those of DI $\eqref{amfld}$ needed later. It also states that for large enough $l\geq1$, the solutions of DI $\eqref{amfld}$  are within an $\epsilon$-neighborhood of the solutions of DI $\eqref{mfld}$ for every initial condition lying in a compact subset of $\mathbb{R}^d$.

\begin{lemma}
\label{sola}
For every $T>0$,
\begin{itemize}
\item [(a)] for every $l\geq1$, for every $x\in\mathbb{R}^d$, $S(T,x)\subseteq S^{(l+1)}(T,x)\subseteq S^{(l)}(T,x)$.
\item [(b)] for every $x\in\mathbb{R}^d$, $S(T,x)=\cap_{l\geq1}S^{(l)}(T,x)$.
\item [(c)] for any $Y\subseteq\mathbb{R}^d$, $S(T,Y)=\cap_{l\geq1}S^{(l)}(T,Y)$.
\item [(d)] for every $Y\subseteq\mathbb{R}^d$ compact, $S(T,Y)$ is a compact subset of $\mathcal{C}([0,T],\mathbb{R}^d)$ (the vector space of $\mathbb{R}^d$-valued continuous functions on $[0,T]$).
\item [(e)] for every $Y\subseteq\mathbb{R}^d$ compact, for every $l\geq1$, $S^{(l)}(T,Y)$ is a compact subset of $\mathcal{C}([0,T],\mathbb{R}^d)$.
\item [(f)] for every $Y\subseteq\mathbb{R}^d$ compact, for every $\epsilon>0$, there exists $l'\geq1$, such that for every $l\geq l'$, for every $\bm{\mathrm{x}}^{(l)}\in S^{(l)}(T,Y)$, there exists $\bm{\mathrm{x}}\in S(T,Y)$, such that $\sup_{t\in[0,T]}\|\bm{\mathrm{x}}(t)-\bm{\mathrm{x}}^{(l)}(t)\|< \epsilon$.
\end{itemize}
\end{lemma}
\begin{proof}
Fix $T>0$.
\begin{itemize}
\item [(a)] Fix $l\geq1$ and $x\in\mathbb{R}^d$. Let $\bm{\mathrm{x}}\in S(T,x)$. Then we have that $\bm{\mathrm{x}}$ is absolutely continuous with $\bm{\mathrm{x}}(0)=x$ and for $a.e.\ t\in[0,T]$, $\frac{d\bm{\mathrm{x}}(t)}{dt}\in F(\bm{\mathrm{x}}(t))$. By Lemma \ref{ctem}$(b)$, we know that for every $t\in [0,T]$, $F(\bm{\mathrm{x}}(t))\subseteq F^{(l+1)}(\bm{\mathrm{x}}(t))$. Therefore for $a.e.\ t\in[0,T]$, $\frac{d\bm{\mathrm{x}}(t)}{dt}\in F^{(l+1)}(\bm{\mathrm{x}}(t))$, from which we get that $\bm{\mathrm{x}}\in S^{(l+1)}(T,x)$.
Hence $S(T,x)\subseteq S^{(l+1)}(T,x)$. Using the fact that for every $x'\in\mathbb{R}^d$, $F^{(l+1)}(x')\subseteq F^{(l)}(x')$ (see Lemma \ref{ctem}$(b)$), a similar argument gives us that $S^{(l+1)}(T,x)\subseteq S^{(l)}(T,x)$.
\item [(b)] Fix $x\in\mathbb{R}^d$. From part $(a)$ of this lemma we have that $S(T,x)\subseteq \cap_{l\geq1}S^{(l)}(T,x)$. Let $\bm{\mathrm{x}}\in \cap_{l\geq1}S^{(l)}(T,x)$. Then $\bm{\mathrm{x}}$ is absolutely continuous with $\bm{\mathrm{x}}(0)=x$ and for every $l\geq1$, for $a.e.\ t\in[0,T]$, $\frac{d\bm{\mathrm{x}}(t)}{dt}\in F^{(l)}(\bm{\mathrm{x}}(t))$. 
Thus for $a.e.\ t\in[0,T]$, for every $l\geq1$, $\frac{d\bm{\mathrm{x}}(t)}{dt}\in F^{(l)}(\bm{\mathrm{x}}(t))$. Hence for $a.e\ t\in[0,T]$, $\frac{d\bm{\mathrm{x}}(t)}{dt}\in \cap_{l\geq1}F^{(l)}(\bm{\mathrm{x}}(t))=F(\bm{\mathrm{x}}(t))$, where the equality follows from Lemma \ref{ctem}$(e)$. Therefore $\bm{\mathrm{x}}\in S(T,x)$, from which we get that $\cap_{l\geq1}S^{(l)}(T,x)\subseteq S(T,x)$.
\item [(c)] Follows from part $(a)$ and $(b)$ of this lemma.
\item [(d) \& (e)] Follows from \cite[Lemma~3.1]{benaim1}.
\item [(f)] Suppose not. Then there exists $Y\subseteq\mathbb{R}^d$ compact and $\epsilon>0$, such that for every $l'\geq1$, there exists $l\geq l'$ and $\bm{\mathrm{x}}^{(l)}\in S^{(l)}(T,Y)$, such that $\bm{d}(\bm{\mathrm{x}}^{(l)},S(T,Y))\geq\epsilon$, where $\bm{d}(\bm{\mathrm{x}}^{(l)},S(T,Y)):=\inf_{\bm{\mathrm{x}}\in S(T,Y)}\sup_{t\in[0,T]}\|\bm{\mathrm{x}}^{(l)}(t)-\bm{\mathrm{x}}(t)\|$. Thus we can obtain a sequence of solutions, say $\{\bm{\mathrm{x}}^{(l_k)}\}_{k\geq1}$, such that for every $k\geq1$, $1\leq l_k<l_{k+1}$ and $\bm{\mathrm{x}}^{(l_k)}\in S^{(l_k)}(T,Y)$ with $\bm{d}(\bm{\mathrm{x}}^{(l_k)},S(T,Y))\geq\epsilon$. From part $(a)$ of this lemma, we have that for every $k\geq1$, $S^{(l_k)}(T,Y)\subseteq S^{(1)}(T,Y)$ and hence $\{\bm{\mathrm{x}}^{(l_k)}\}_{k\geq1}\subseteq S^{(1)}(T,Y)$. Since $Y\subseteq\mathbb{R}^d$ is compact, by part $(e)$ of this lemma we know that $S^{(1)}(T,Y)$ is a compact subset of $\mathcal{C}([0,T],\mathbb{R}^d)$. Thus there exists a subsequence of $\{\bm{\mathrm{x}}^{(l_k)}\}_{k\geq1}$, say $\{\bm{\mathrm{x}}^{(l_{k_j})}\}_{j\geq1}$ such that $\bm{\mathrm{x}}^{(l_{k_j})}\to \bm{\mathrm{x}}^*$ as $j\to\infty$ in $\mathcal{C}([0,T],\mathbb{R}^d)$ and $\bm{\mathrm{x}}^*\in S^{(1)}(T,Y)$. Since for every $j\geq1$, $\bm{d}(\bm{\mathrm{x}}^{(l_{k_j})},S(T,Y))\geq\epsilon$, we get that $\bm{d}(\bm{\mathrm{x}}^*,S(T,Y))\geq\epsilon$ and hence $\bm{\mathrm{x}}^*\notin S(T,Y)$. 
From part $(a)$ of this lemma, we get that for every $l\geq1$, for $J:=\min\{j\geq1:l_{k_j}\geq l\}$, $\{\bm{\mathrm{x}}^{(l_{k_j})}\}_{j\geq J}\subseteq S^{(l)}(T,Y)$. Further by part $(e)$ of this lemma we have that for every $l\geq 1$, $S^{(l)}(T,Y)$ is a compact subset of $\mathcal{C}([0,T],\mathbb{R}^d)$. Thus for every $l\geq1$, $\bm{\mathrm{x}}^*\in S^{(l)}(T,Y)$ and hence $\bm{\mathrm{x}}^*\in\cap_{l\geq1}S^{(l)}(T,Y)=S(T,Y)$ (see part $(c)$ of this lemma). This leads to a contradiction.\qed 
\end{itemize}
\end{proof}

The part $(f)$ of the above lemma provides the necessary approximation result. Further since the set-valued maps $F^{(l)}$ admit a single-valued parametrization ($f^{(l)}$ as in Lemma \ref{param}), a solution of DI $\eqref{amfld}$ can be viewed as a solution of the ordinary differential equation (o.d.e.) given by,
\begin{equation}
\label{pmfld}
\frac{dx}{dt}=f^{(l)}(x,u(t)), 
\end{equation}
 for some $u:[0,\infty)\rightarrow U$ measurable and vice versa. The lemma below summarizes some useful results on the solutions of o.d.e. $\eqref{pmfld}$ and its vector field.
 
 \begin{lemma}\label{pvf} For every $l\geq1$,
 \begin{itemize}
 \item [(a)]  for every $T>0$, for any $u:[0,T]\rightarrow U$ measurable, for every initial condition, the set of solutions of o.d.e. $\eqref{pmfld}$ is non-empty. That is, for every $x_0\in\mathbb{R}^d$, there exists $\bm{\mathrm{x}}:[0,T]\rightarrow \mathbb{R}^d$ such that, $\bm{\mathrm{x}}$ is absolutely continuous, $\bm{\mathrm{x}}(0)=x_0$ and for $a.e.\ t\in [0,T]$, $\frac{d\bm{\mathrm{x}}(t)}{dt}=f^{(l)}(\bm{\mathrm{x}}(t),u(t))$.
 
 \item [(b)] for every $T>0$, for every $Y\subseteq\mathbb{R}^d$ compact, there exists $C_1(Y,T,l)>0$, such that for every $u:[0,T]\rightarrow U$ measurable, every solution of o.d.e. $\eqref{pmfld}$ with initial condition in $Y$, say $\bm{\mathrm{x}}:[0,T]\rightarrow\mathbb{R}^d$, satisfies,
 \begin{equation*}
 \sup_{t\in[0,T]}\|\bm{\mathrm{x}}(t)\|\leq C_1(Y,T,l).
 \end{equation*} 
 
 \item [(c)] for any $Y\subseteq\mathbb{R}^d$ compact, there exists $L(Y,l)>0$, such that for every $T>0$, for every $u:[0,T]\rightarrow U$, the map $h:Y\times[0,T]\rightarrow\mathbb{R}^d$, given by $h(x,t):=f^{(l)}(x,u(t))$ for every $(x,t)\in Y\times [0,T]$, satisfies,
 \begin{equation*}
 \|h(x,t)-h(x',t)\|\leq L(Y,l)\|x-x'\|,
 \end{equation*}      
 for every $x,\ x'\in Y$ and for every $t\in[0,T]$.
 \item [(d)] for every $T>0$, for every $u:[0,T]\rightarrow U$, for every initial condition, o.d.e. $\eqref{pmfld}$ admits a unique solution.
\end{itemize}
\end{lemma}
\begin{proof} Fix $l\geq1$.
\begin{itemize}
\item [(a)] Fix $T>0$ and $u:[0,T]\rightarrow U$ measurable. The proof of this part is a direct application of \cite[Thm.~3.8]{regan}. We show here that the sufficient conditions required to apply the said theorem are satisfied by the vector field of the o.d.e. $\eqref{pmfld}$. First, we show that $f^{(l)}(\cdot,u(\cdot))$ is a Caratheodary function (see \cite[Defn.~3.2]{regan}). By Lemma \ref{param}, it is clear that for every $t\in[0,T]$, the map $x\rightarrow f^{(l)}(x,u(t))$ is continuous and for every $x\in\mathbb{R}^d$, the map $t\rightarrow f^{(l)}(x,u(t))$ is measurable. Further by Lemma \ref{param}$(b)$, we have that for any $c>0$, for every $x\in\mathbb{R}^d$ with $\|x\|\leq c$, for every $t\in[0,T]$, $\|f^{(l)}(x,t)\|\leq K^{(l)}(1+c)$. Thus $f^{(l)}(\cdot,u(\cdot))$ is a Caratheodary function. Final condition to verify is on the rate of growth of solutions. By Lemma \ref{param}$(b)$, $\|f^{(l)}(x,u(t))\|\leq \psi(\|x\|):=K^{(l)}(1+\|x\|)$. The function $\psi:[0,\infty)\rightarrow [0,\infty)$, is clearly positive everywhere and the function $\frac{1}{\psi}$ is locally integrable on $[0,\infty)$. A simple argument gives us that  for every $r>0$, the integral $\int_{r}^{\infty}\frac{d\tilde{r}}{\psi(\tilde{r})}$ can be lower bounded by the tail of $\frac{1}{K^{(l)}}\sum_{n=1}^{\infty}\frac{1}{n}$. Hence for every $r>0$, $\int_{r}^{\infty}\frac{d\tilde{r}}{\psi(\tilde{r})}=\infty$. Now \cite[Thm.~3.8]{regan} can be applied to obtain the required result.
\item [(b)] Fix $T>0$ and $Y\subseteq\mathbb{R}^d$ compact. Since $Y$ is compact, there exists $r>0$ such that $\sup_{y\in Y}\|y\|\leq r$. Set $C_1(Y,T,l):=(r+K^{(l)}T)e^{K^{(l)}T}$, where $K^{(l)}>0$ is as in Lemma \ref{param}$(b)$. For some $u:[0,T]\rightarrow U$ measurable and for some $x_0\in Y$, let $\bm{\mathrm{x}}:[0,T]\rightarrow \mathbb{R}^d$ be a solution of o.d.e. $\eqref{pmfld}$ with initial condition $x_0$. Then, for every $t\in[0,T]$, $\bm{\mathrm{x}}(t)=x_0+\int_{0}^{t}f^{(l)}(x(s),u(s))ds$ and hence for every $t\in[0,T]$
\begin{align}
\|\bm{\mathrm{x}}(t)\|&\leq \|x_0\|+\int_{0}^{t}\|f^{(l)}(\bm{\mathrm{x}}(s),u(s))\|ds\nonumber\\
                                      \label{tmp4}
                                     &\leq r+K^{(l)}T+K^{(l)}\int_{0}^{t}\|\bm{\mathrm{x}}(s)\|ds
\end{align}
where, $\eqref{tmp4}$ follows from the fact that $x_0\in Y$ and Lemma \ref{param}$(b)$. The required bound follows from $\eqref{tmp4}$ and Gronwall's result (see \cite[Sec.~11.2.1, Lemma~6]{borkartxt}).                                   
\item [(c)] Fix $Y\subseteq\mathbb{R}^d$ compact. It is enough to show that there exists $L(Y,l)>0$, such that for every $y_1,\ y_2\in Y$, $\sup_{u\in U}\|f^{(l)}(y_1,u)-f^{(l)}(y_2,u)\|\leq L(Y,l)\|y_1-y_2\|$. From Lemma \ref{param}$(c)$, we know that for every $x_0\in Y$, there exists $\delta(x_0,l)>0$ and $L(x_0,l)>0$, such that for every $x,x'\in x_0+\delta(x_0,l)U$, for every $u\in U$, $\|f^{(l)}(y_1,u)-f^{(l)}(y_2,u)\|\leq L(Y,l)\|y_1-y_2\|$. Let $\mathcal{G}:=\{x_0+\frac{\delta(x_0,l)}{2}\mathring{U}:x_0\in Y\}$, where, $\mathring{U}$ denotes the interior of $U$. Since $Y$ is compact and $\mathcal{G}$ is an open cover of $Y$, there exists $\{x_1,x_2,\dots,x_k\}\subseteq Y$, such that $Y\subseteq\cup_{i=1}^{k}(x_i+\frac{\delta(x_i,l)}{2}\mathring{U})$. Set $\delta(Y,l):=\min_{1\leq i\leq k}\frac{\delta(x_i,l)}{2}$ and $L_0(Y,l):=\max_{1\leq i\leq k}L(x_i,l)$.  

Let $y_1,y_2\in (Y\times Y)\cap\{(y_1,y_2): \|y_1-y_2\|<\delta(Y,l)\}$. Then we know that there exists $i\in\{1,\dots,k\}$, such that $y_1\in x_i+\frac{\delta(x_i,l)}{2}\mathring{U}$. Further since $\|y_1-y_2\|<\delta(Y,l)\leq \frac{\delta(x_i,l)}{2}$, we have that $y_2\in x_i +\delta(x_i,l)\mathring{U}$. Therefore $y_1,y_2\in x_i+\delta(x_i,l)U$ and hence, for every $u\in U$, $\|f^{(l)}(y_1,u)-f^{(l)}(y_2,u)\|\leq L(x_i,l)\|y_1-y_2\|\leq L_0(Y,l)\|y_1-y_2\|$. Thus for every $y_1,y_2\in (Y\times Y)\cap\{(y_1,y_2): \|y_1-y_2\|<\delta(Y,l)\}$, $\sup_{u\in U}\|f^{(l)}(y_1,u)-f^{(l)}(y_2,u)\|\leq L_0(Y,l)\|y_1-y_2\|$.     

Let $E:=(Y\times Y)\cap \{(y_1,y_2):\|y_1-y_2\|\geq\delta(Y,l)\}$. By Lemma \ref{param}$(b)$, the map $(y_1,y_2)\in Y\times Y\rightarrow \sup_{u\in U}\|f^{(l)}(y_1,u)-f^{(l)}(y_2,u)\|$ is well defined. Further using the fact that for every $(y_1,y_2), (y_1',y_2')\in Y\times Y$, $|\sup_{u\in U}\|f^{(l)}(y_1,u)-f^{(l)}(y_2,u)\|-\sup_{u\in U}\|f^{(l)}(y_1',u)-f^{(l)}(y_2',u)\||\leq\sup_{u\in U}\|f^{(l)}(y_1,u)-f^{(l)}(y_1',u)\|+\sup_{u\in U}\|f^{(l)}(y_2,u)-f^{(l)}(y_2',u)\|$ and Lemma \ref{param}$(c)$, we have that the map $(y_1,y_2)\rightarrow\sup_{u\in U}\|f^{(l)}(y_1,u)-f^{(l)}(y_2,u)\|$ is continuous. Thus the map $(y_1,y_2)\in E\rightarrow \frac{\sup_{u\in U}\|f^{(l)}(y_1,u)-f^{(l)}(y_2,u)\|}{\|y_1-y_2\|}$ is a continuous function on a compact set $E$ and hence achieves a maximum, say $L_1(Y,l)\geq0$. Therefore for every $(y_1,y_2)\in (Y\times Y)\cap \{(y_1,y_2):\|y_1-y_2\|\geq\delta(Y,l)\}$, $\sup_{u\in U}\|f^{(l)}(y_1,u)-f^{(l)}(y_2,u)\|\leq L_1(y,l)\|y_1-y_2\|$.

Thus from the arguments in the two preceding paragraphs we have that there exists $L(Y,l):=\max\{L_0(Y,l),L_1(Y,l)\}$, such that, for every $y_1,y_2\in Y$, $\sup_{u\in U}\|f^{(l)}(y_1,u)-f^{(l)}(y_2,u)\|\leq L(Y,l)\|y_1-y_2\|$.  
\item [(d)] Using parts $(b)$ and $(c)$ of this lemma, the proof of uniqueness follows from arguments similar to that of \cite[Thm.~3.4]{regan}.\qed                    
\end{itemize}
\end{proof}
%----------------------------------------------------------------------------------------------------------------------------------------------------------------------------------------------------------------------------------

\subsection{Bounding procedure}
\label{prf_lip3}
In this section we show that the lower bound on the probability of the event that the iterates converge to an attracting set given that after a large number of iterations the iterates lies in a neighborhood of it depends mainly on the additive noise terms. 

In order to accomplish this we first define some terms which are a measure of the distance of the linearly interpolated trajectory of recursion $\eqref{rec}$, that is $\bar{X}$ (see $\eqref{lit})$ to the solutions of the DI $\eqref{mfld}$ over a $T>0$ length time interval, among others. Recall from section \ref{mr} that $\mathcal{O}'\subseteq \mathbb{R}^d$, is an open neighborhood of the attracting set $A$ (as in assumption $(A4)$) with compact closure, such that $A\subseteq \mathcal{O}'\subseteq \bar{\mathcal{O}'}\subseteq \mathcal{O}$, where $\mathcal{O}$ denotes the fundamental neighborhood of $A$. Thus we can find an $\epsilon_0>0$, such that $N^{\epsilon_0}(\bar{\mathcal{O}'})\subseteq \mathcal{O}$ and $N^{2\epsilon_0}(A)\subseteq \mathcal{O}'$, where for any $\epsilon>0$, $N^{\epsilon}(\cdot)$ denotes the $\epsilon$-neighborhood of a set. Further, since $A$ is an attracting set for the flow of DI $\eqref{mfld}$, for $\epsilon_0>0$ as obtained above, there exists $T_A>0$, such that for every $x\in \mathcal{O}$, for every $t\geq T_A$, $\Phi(t,x)\in N^{\epsilon_0}(A$.
Throughout the rest of this paper $\epsilon_0$ and $T_A$ will denote the constants as obtained above.

For every $T>0$, for every $n\geq0$, 
\begin{enumerate}[label=\textbf{Definition~\arabic*},wide, itemsep =0pt, topsep=0pt]
\item \label{d1} : let $\tau(n,T):=\min\{k\geq n: t(k)\geq t(n)+T\}$, where $t(n)$, for every $n\geq0$ are as defined in section \ref{mtvtn}. That is $\tau(n,T)$ denotes the first iterate such that, at least time $T$ has elapsed since the $n^{th}$ iteration. Further the time elapsed from iteration $n$ to iteration $\tau(n,T)$, be denoted by $\Delta(n,T)$, that is $\Delta(n,T):=t(\tau(n,T))-t(n)$. Then by the choice of our step sizes we have that $T\leq \Delta(n,T)\leq T+1$.
\item \label{d2} : for every $\omega\in \Omega$,  $\rho(\omega,n,T):=\inf_{\bm{\mathrm{x}}\in S(T,\bar{\mathcal{O}'})}\sup_{t\in[0,T]}\|\bar{X}(\omega,t+t(n))-\bm{\mathrm{x}}(t)\|$, where $S(T,\bar{\mathcal{O}'})$ denotes the set of solutions of DI $\eqref{mfld}$ as defined in equation $\eqref{ss0}$.
\item \label{d3} : for every $\omega \in \Omega$, for every $l\geq1$, let $\bar{\bm{\mathrm{x}}}^{(l)}(\cdot;n,T,\omega):[0,T]\rightarrow\mathbb{R}^d$ denote the unique solution of the  o.d.e. 
\begin{equation}
\label{pmfld1}
\frac{dx}{dt}=f^{(l)}(x,u(t;n,T,\omega)),
\end{equation}
with initial condition $\bar{\bm{\mathrm{x}}}^{(l)}(0;n,T,\omega)=X_n(\omega)$, where $u(\cdot;n,T,\omega):[0,T]\rightarrow U$ is defined such that, for every $t\in[0,T]$, $u(t;n,T,\omega):=U_k^{(l)}(\omega)$, where $U_k^{(l)}$ is as in equation $\eqref{rec1}$ and $k$ is such that $t+t(n)\in [t(k),t(k+1))$ (for a proof of existence and uniqueness of solutions to o.d.e. $\eqref{pmfld1}$, see Lemma \ref{pvf}).  It is easy to see that for every $l\geq1$, $\bar{\bm{\mathrm{x}}}^{(l)}(\cdot;n,T,\omega)\in S^{(l)}(T,X_n(\omega))$, where $S^{(l)}(T,X_n(\omega))$ denotes the set of solutions of DI $\eqref{amfld}$, as defined in $\eqref{ss1}$.
\item \label{d4} : for every $\omega\in \Omega$, for every $l\geq1$, $\rho^{(l)}_1(\omega,n,T):=\sup_{t\in[0,T]}\|\bar{X}(\omega,t+t(n))-\bar{\bm{\mathrm{x}}}^{(l)}(t;n,T,\omega)\|$ and $\rho^{(l)}_2(\omega,n,T):=\inf_{\bm{\mathrm{x}}\in S(T,\bar{\mathcal{O}'})}\sup_{t\in[0,T]}\|\bar{\bm{\mathrm{x}}}^{(l)}(t;n,T,\omega)-\bm{\mathrm{x}}(t)\|$.
\item \label{d5} : for any $T_u\geq T_A$, for any $n_0\geq0$, let $\{n_m\}_{m\geq1}$ denote a subsequence of natural numbers defined such that for every $m\geq0$, $T_A\leq T_m:= t(n_{m+1})-t(n_m)\leq T_u$.   
\end{enumerate}  

Now we collect sample paths of interest using the quantities $\rho$, $\rho_1^{(l)}$ and $\rho_2^{(l)}$. The next lemma summarizes results in this regard.
\begin{lemma}
\label{events}
For every $T_u\geq T_A$, for every $n_0\geq0$, for every $l\geq1$, for every event $E\in \mathscr{F}_{n_0}$, such that $E\subseteq \{\omega : X_{n_0}(\omega)\in \mathcal{O}'\}$, for every $\{n_m\}_{m\geq1}$ as in \ref{d5},
\begin{itemize}
\item [(a)] for every $M\geq0$,
\begin{small} 
\begin{align*}
E \cap \left(\cap_{m=0}^{M}\left\{\omega\in\Omega:\rho_1^{(l)}(\omega,n_m,T_m)+\rho_2^{(l)}(\omega,n_m,T_m)<\epsilon_0\right\}\right)&\subseteq E \cap \left(\cap_{m=0}^{M}\left\{\omega\in\Omega:\rho(\omega,n_m,T_m)<\epsilon_0\right\}\right)\\
&\subseteq \left\{\omega\in\Omega: X_{n_{M+1}}(\omega)\in\mathcal{O}'\right\},
\end{align*}
\end{small}
\item [(b)]
\begin{small}
\begin{align*}
\mathbb{P}\!\left(\!E \cap\! \left(\cap_{m\geq0}\!\left\{\omega\!\in\!\Omega:\rho_1^{(l)}(\omega,n_m,T_m)+\rho_2^{(l)}(\omega,n_m,T_m)<\epsilon_0\!\right\}\!\right)\!\right)\!&\leq \mathbb{P}\!\left(E \cap\! \left(\cap_{m\geq0}\!\left\{\omega\!\in\!\Omega:\rho(\omega,n_m,T_m)\!<\epsilon_0\!\right\}\right)\right)\\
&\leq\mathbb{P}\left(E\cap\left\{\omega\in\Omega: X_{n}(\omega)\to A\ \text{as }n\to\infty\right\}\right),
\end{align*}
\end{small}

\end{itemize}
where, $\{T_m\}_{m\geq0}$ is as in \ref{d5}.
\end{lemma}
\begin{proof}
Fix $n_0\geq0$, $l\geq1$ and $E\in\mathscr{F}_{n_0}$, such that $E\subseteq \{\omega\in\Omega:X_{n_0}(\omega)\in \mathcal{O}'\}$.
\begin{itemize}
\item [(a)] For every $m\geq0$, for every $\omega\in\Omega$, from \ref{d2} and \ref{d4}, it is clear that,
\begin{equation*}
\rho(\omega,n_m,T_m)\leq\rho_1^{(l)}(\omega,n_m,T_m)+\rho_2^{(l)}(\omega,n_m,T_m),
\end{equation*}
from which we get that for every $m\geq0$, 
\begin{equation*}
\{\omega\in\Omega:\rho_1^{(l)}(\omega,n_m,T_m)+\rho_2^{(l)}(\omega,n_m,T_m)<\epsilon_0\}\subseteq\{\omega\in\Omega:\rho(\omega,n_m,T_m)<\epsilon_0\}.
\end{equation*}
Therefore,
\begin{small}
\begin{equation*}
E \cap \left(\cap_{m=0}^{M}\left\{\omega\in\Omega:\rho_1^{(l)}(\omega,n_m,T_m)+\rho_2^{(l)}(\omega,n_m,T_m)<\epsilon_0\right\}\right)\subseteq E \cap \left(\cap_{m=0}^{M}\left\{\omega\in\Omega:\rho(\omega,n_m,T_m)<\epsilon_0\right\}\right).
\end{equation*}
\end{small}
The proof of the second inclusion follows from induction. Fix $M=0$ and $\omega\in E\cap \{\omega\in\Omega:\rho(\omega,n_0,T_0)<\epsilon_0\}$. Then $X_{n_0}(\omega)\in \mathcal{O}'$. Since $T_0\geq T_A$, we have that for every $\bm{\mathrm{x}}\in S(T_0,\bar{\mathcal{O}'})$, $\bm{\mathrm{x}}(T_0)\in N^{\epsilon_0}(A)$. Further, since $\rho(\omega,n,T_0)<\epsilon_0$ and by Lemma \ref{sola}$(d)$, we get that there exists $\bm{\mathrm{x}}\in S(T_0,\bar{\mathcal{O}'})$, such that $\|\bar{X}(\omega,t(n_1))-\bm{\mathrm{x}}(T_0)\|=\|X_{n_1}(\omega)-\bm{\mathrm{x}}(T_0)\|<\epsilon_0$ and hence $X_{n_1}(\omega)\in N^{2\epsilon_0}(A)\subseteq \mathcal{O}'$. Therefore $\omega\in \{\omega\in\Omega: X_{n_1}(\omega)\in \mathcal{O}'\}$. Thus the inclusion is true for $M=0$. Suppose the inclusion is true for some $M>0$. Let $\omega\in E \cap \left(\cap_{m=0}^{M+1}\left\{\omega\in\Omega:\rho(\omega,n_m,T_m)<\epsilon_0\right\}\right)$. Since the inclusion is true for $M$, we have that $X_{n_{M+1}}(\omega)\in\mathcal{O}'$. Now by arguments exactly same as those for the base case (that is for $M=0$) we get that $X_{n_{M+2}}(\omega)\in\mathcal{O}'$. Therefore the inclusion is true for $M+1$.
\item [(b)] The first inequality follows from part $(a)$ of this lemma. We shall provide a proof of the second inequality. Let $\omega\in E \cap \left(\cap_{m\geq0}\left\{\omega\in\Omega:\rho(\omega,n_m,T_m)<\epsilon_0\right\}\right)$. Then by part $(a)$ of this lemma we have that for every $m\geq0$, $X_{n_m}(\omega)\in \mathcal{O}'$. Since $\bar{\mathcal{O}'}$ is compact, by Lemma \ref{sola}$(d)$, we have that, $S(T_u,\bar{\mathcal{O}'})$ is a compact subset of $\mathcal{C}([0,T_u],\mathbb{R}^d)$, and hence there exists $C(\bar{\mathcal{O}'},T_m)>0$ such that, $\sup_{\bm{\mathrm{x}}\in S(T_u,\bar{\mathcal{O}'})}\sup_{t\in[0,T_u]}\|\bm{\mathrm{x}}(t)\|\leq C(\bar{\mathcal{O}'},T_m)$. Further since for every $m\geq0$, $T_m\leq T_u$, we get that $\sup_{\bm{\mathrm{x}}\in S(T_m,\bar{\mathcal{O}'})}\sup_{t\in[0,T_m]}\|\bm{\mathrm{x}}(t)\|\leq \sup_{\bm{\mathrm{x}}\in S(T_u,\bar{\mathcal{O}'})}\sup_{t\in[0,T_u]}\|\bm{\mathrm{x}}(t)\|\leq C(\bar{\mathcal{O}'},T_m)$. By our choice of $\omega$, we have that for every $m\geq0$, $\rho(\omega,n_m,T_m)<\epsilon_0$ and by \ref{d2}, we get that for every $m\geq0$, $\sup_{t\in [0,T_m]}\|\bar{X}(\omega,t+t(n_m))\|\leq C(\bar{\mathcal{O}'},T_m)+\epsilon_0$. Therefore $\omega$, is such that $\sup_{n\geq0}\|X_n(\omega)\|<\infty$ and for every $m\geq0$, $X_{n_m}(\omega)\in\mathcal{O}'$. Thus $\lambda(\omega)$ (see equation $\eqref{ls}$ for definition), is non-empty, compact and $\lambda(\omega)\cap \bar{\mathcal{O}'}\subseteq \lambda(\omega)\cap B(A)\neq \emptyset$, where $B(A)$ denotes the basin of attraction of the attracting set $A$. By  Theorem \ref{aptls}$(a)$, we have that for almost every $\omega$ in $ E \cap \left(\cap_{m\geq0}\left\{\omega\in\Omega:\rho(\omega,n_m,T_m)<\epsilon_0\right\}\right)$ the iterates converge to the attracting set $A$. Therefore we get that $\mathbb{P}\left( E \cap \left(\cap_{m\geq0}\left\{\omega\in\Omega:\rho(\omega,n_m,T_m)<\epsilon_0\right\}\right)\right)\leq \mathbb{P}\left(E\cap\{\omega\in\Omega: X_n(\omega)\to A\ \text{as }n\to \infty)\}\right)$.\qed
\end{itemize}
\end{proof}   

The quantity $\rho_1^{(l)}$ as in \ref{d4}, captures the difference between the linearly interpolated trajectory of recursion $\eqref{rec}$ and the solution of the o.d.e. $\eqref{pmfld1}$ over a $T>0$ length time interval. This difference can be shown to comprise of two components namely, the error due to discretization and the error due to additive noise terms. By the step size assumption, that is $(A2)$, we know that the step sizes are converging to zero. Hence intuition suggests that after a large number of iterations have elapsed the discretization error must be negligible and the contribution to the difference term $\rho_1^{(l)}$ is mainly due to the additive noise terms. The following is made precise in the lemma below. A brief outline of the proof of this lemma which follows from Lemma \ref{pvf}$(c)$ and \cite[Ch.~2, Lemma~1]{borkartxt}, is presented in Appendix \ref{prf_int}.

\begin{lemma}
\label{bd0}
For every $l\geq1$, for every $T_u\geq T_A$, there exists $N_0'\geq1$, such that for every $n_0\geq N_0'$, for every $E\in\mathscr{F}_{n_0}$ such that, $E\subseteq \{\omega\in\Omega:X_{n_0}(\omega)\in\mathcal{O}'\}$, for every sequence $\{n_m\}_{m\geq0}$ as in \ref{d5}, for every $m\geq0$, we have,
%\begin{small}
\begin{equation*}
\mathbb{P}\!\left(\mathcal{B}_{m-1}^{(l)}\cap\{\omega\!\in\! \Omega:\rho_{1}^{(l)}(\omega,n_m,T_m)\geq \frac{\epsilon_0}{2}\}\!\right)\!\leq
\mathbb{P}\left(\!\{\omega\!\in\!\Omega:\!\!\!\max_{n_m\leq j\leq n_{m+1}}\!\!\!\|\zeta_j(\omega)-\zeta_{n_m}(\omega)\|\geq \frac{\epsilon_0}{4K_0(T_u)}\}\!\cap\mathcal{B}^{(l)}_{m-1} \!\right),
\end{equation*}
%\end{small}
where,
\begin{itemize}
\item $\mathcal{B}_{-1}^{(l_0)}:=E$ and for every $M\geq0$, $\mathcal{B}_{M}^{(l_0)}\!:=\!E\cap\left(\cap_{m=0}^{M}\{\omega\!\in\!\Omega:\rho_1^{(l_0)}(\omega,n_m,T_m)+\rho_2^{(l_0)}(\omega,n_m,T_m)\!<\!\epsilon_0\}\!\right)$, 
\item for every $j\geq1$, $\zeta_j:=\sum_{n=0}^{j-1}a(n)M_{n+1}$, where $\{M_{n}\}_{n\geq1}$ denote the additive noise terms as defined in assumption $(A3)$,
\item $\{T_m\}_{m\geq0}$ is as in \ref{d5} and $K_0(T_u)>0$ is a positive constant increasing in $T_u$
\end{itemize}
\end{lemma}

Suppose event $E$ as in the lemma above occurs with some positive probability. Then the next lemma says that the lower bound of $\mathbb{P}\left(\{\omega\in\Omega : X_n(\omega)\to A\ \text{as }n\to\infty\}|E\right)$ depends mainly on the additive noise terms for $n_0$ large. 

\begin{lemma}
\label{bd1}
For every $T_u\geq T_A$, there exists $l_0\geq1$ and $N_0'\geq1$, such that for every $n_0\geq N_0'$, for every $E\in\mathscr{F}_{n_0}$ such that, $E\subseteq \{\omega\in\Omega:X_{n_0}(\omega)\in\mathcal{O}'\}$ and $\mathbb{P}(E)>0$, for every sequence $\{n_m\}_{m\geq0}$ as in \ref{d5}, we have,
\begin{equation}
\label{mbnd0}
\mathbb{P}\left(\{\omega\in\Omega : X_n(\omega)\to A\ \text{as }n\to\infty\}|E\right)\geq1-\sum_{m=0}^{\infty}\mathbb{P}\left(\max_{n_m\leq j\leq n_{m+1}}\|\zeta_j-\zeta_{n_m}\|\geq \frac{\epsilon_0}{4K_0(T_u)}|\mathcal{B}^{(l_0)}_{m-1}\right),
\end{equation}
where, the sequence of events $\{\mathcal{B}_m^{(l_0)}\}_{m\geq-1}$, the sequence of random vectors $\{\zeta_j\}_{j\geq1}$ and the constant $K_0(T_u)$ are as defined in Lemma \ref{bd0}.
\end{lemma}
\begin{proof} By Lemma \ref{sola}$(f)$, we get that there exists $l_0\geq1$ (depending on $\bar{\mathcal{O}'}$, $T_u$ and $\epsilon_0$) such that for every $\bm{\mathrm{x}}^{(l_0)}\in S^{(l_0)}(T_u,\bar{\mathcal{O}'})$, there exists $\bm{\mathrm{x}}\in S(T_u,\bar{\mathcal{O}'})$ such that $\sup_{t\in[0,T_u]}\|\bm{\mathrm{x}}^{(l_0)}(t)-\bm{\mathrm{x}}(t)\|<\frac{\epsilon_0}{2}$. Further by Lemma \ref{events}$(a)$ and definition of $E$, we get that for every $m\geq0$, $\mathcal{B}_{m-1}^{(l_0)}\subseteq\{\omega\in\Omega:X_{n_{m}}(\omega)\in \mathcal{O}'\}$. Therefore, for every $\omega\in \mathcal{B}_{m-1}^{(l_0)}$, $\bar{\bm{\mathrm{x}}}^{(l_0)}(\cdot;n_m,T_m,\omega)\in S^{(l_0)}(T_m,\bar{\mathcal{O}'})$ and
\begin{align}
\rho_2^{(l_0)}(\omega,n_m,T_m)&=\inf_{\bm{\mathrm{x}}\in S(T_m,\bar{\mathcal{O}'})}\sup_{t\in[0,T_m]}\|\bar{\bm{\mathrm{x}}}^{(l_0)}(t;n_m,T_m, 							\omega)-\bm{\mathrm{x}}(t)\|\nonumber\\
\label{tmp5}
                                                      &\leq\inf_{\bm{\mathrm{x}}\in S(T_u,\bar{\mathcal{O}'})}\sup_{t\in[0,T_u]}\|\bar{\bm{\mathrm{x}}}^{(l_0)}(t;n,T_u,\omega)-                                                         					\bm{\mathrm{x}}(t)\|\\
                                                      \label{tmp6}
                                                      &<\frac{\epsilon_0}{2},
\end{align} 
where $\eqref{tmp5}$ follows from the fact that $T_m\leq T_u$ and $\eqref{tmp6}$ follows from our choice of $l_0$ and \ref{d3}. Therefore for every $m\geq-1$,
\begin{equation*}
\mathcal{B}_m^{(l_0)}\cap\{\omega\in\Omega:\rho_1^{(l_0)}(\omega,n_{m+1},T_m)+\rho_2^{(l_0)}(\omega,n_{m+1},T_m)\geq\epsilon_0\}\subseteq\mathcal{B}_m^{(l_0)}\cap\{\omega\in\Omega:\rho_1^{(l_0)}(\omega,n_{m+1},T_m)\geq\frac{\epsilon_0}{2}\},
\end{equation*}
and hence, 
\begin{small}
\begin{equation}
\label{tmp7}
\mathbb{P}\left(\{\omega\in\Omega:\rho_1^{(l_0)}(\omega,n_{m+1},T_m)+\rho_2^{(l_0)}(\omega,n_{m+1},T_m)\geq\epsilon_0\}|\mathcal{B}_m^{(l_0)}\right)\leq\mathbb{P}\left(\{\omega\in\Omega:\rho_1^{(l_0)}(\omega,n_{m+1},T_m)\geq\frac{\epsilon_0}{2}\}|\mathcal{B}_m^{(l_0)}\right).
\end{equation} 
\end{small}
By Lemma \ref{bd0}, we know that there exists $N_0'\geq1$ such that, for every $n_0\geq N_0'$, for every $m\geq0$, 
\begin{equation*}
\mathbb{P}\!\left(\!\mathcal{B}_{m-1}^{(l_0)}\cap\!\{\omega\!\in\!\Omega:\rho_1^{(l_0)}(\omega,n_{m},T_m)\geq\frac{\epsilon_0}{2}\}\!\right)\!\leq\mathbb{P}\!\left(\!\mathcal{B}_{m-1}^{(l_0)}\cap\!\{\omega\!\in\!\Omega:\!\!\!\!\max_{n_m\leq j\leq n_{m+1}}\!\!\!\|\zeta_j(\omega)-\zeta_{n_m}(\omega)\|\geq \frac{\epsilon_0}{4K_0(T_u)}\}\!\right),
\end{equation*} 
from which it follows that,
\begin{equation}
\label{tmp8}
\mathbb{P}\!\left(\!\{\omega\!\in\!\Omega:\rho_1^{(l_0)}(\omega,n_{m},T_m)\geq\frac{\epsilon_0}{2}\}|\mathcal{B}_{m-1}^{(l_0)}\!\right)\!\leq\mathbb{P}\!\left(\!\{\omega\!\in\!\Omega:\!\!\!\max_{n_m\leq j\leq n_{m+1}}\!\!\!\|\zeta_j(\omega)-\zeta_{n_m}(\omega)\|\geq \frac{\epsilon_0}{4K_0(T_u)}\}|\mathcal{B}_{m-1}^{(l_0)}\!\right).
\end{equation}
For $l_0\geq1$ as obtained above and for $n_0\geq N_0'$, we have that,
\begin{align}
\label{tmp9}
\mathbb{P}\left(X_n\to A\ \text{as }n\to\infty|E\right)&\geq \mathbb{P}\left(\cap_{m\geq0}\{\omega\in\Omega:\rho_1^{(l_0)}(\omega,n_m,T_m)+\rho_2^{(l_0)}(\omega,n_m,T_m)< 										   \epsilon_0\}|E\right)\\
                                                                                      &=1-\mathbb{P}\left(\cup_{m\geq0}\{\omega\in\Omega:\rho_1^{(l_0)}(\omega,n_m,T_m)+\rho_2^{(l_0)}(\omega,n_m,T_m)\geq 										   \epsilon_0\}|E\right)\nonumber\\
                                                                                      &=1-\mathbb{P}\left(\{\omega\in\Omega:\rho_1^{(l_0)}(\omega,n_0,T_m)+\rho_2^{(l_0)}(\omega,n_0,T_m)\geq 										                        \epsilon_0\}|\mathcal{B}^{(l_0)}_{-1}\right)\nonumber\\
                                                                                      \label{tmp10}
                                                                                       -\sum_{m=1}^{\infty}&\mathbb{P}\left(\{\omega\in\Omega:\rho_1^{(l_0)}(\omega,n_m,T_m)+\rho_2^{(l_0)}(\omega,n_m,T_m)\geq                								\epsilon_0\}|\mathcal{B}_{m-1}^{(l_0)}\right)\mathbb{P}\left(\mathcal{B}_{m-1}^{(l_0)}|\mathcal{B}_{-1}^{(l_0)}\right)\\
                                                                                      \label{tmp11}
                                                                                      &\geq1- \sum_{m=0}^{\infty}\mathbb{P}\left(\{\omega\in\Omega:\rho_1^{(l_0)}(\omega,n_m,T_m)+\rho_2^{(l_0)}(\omega,n_m,T_m)									\geq\epsilon_0\}|\mathcal{B}_{m-1}^{(l_0)}\right),
\end{align} 
where, $\eqref{tmp9}$ follows from Lemma \ref{events}$(b)$, $\eqref{tmp10}$ follows from the observation that,
\begin{align*}
\big(\cup_{m\geq0}\{\omega\in\Omega:\rho_1^{(l_0)}(\omega,n_m,T_m)&+\rho_2^{(l_0)}(\omega,n_m,T_m)\geq\epsilon_0\}\big)\cap E=\\
&\cup_{m\geq0}\left(\{\omega\in\Omega:\rho_1^{(l_0)}(\omega,n_m,T_m)+\rho_2^{(l_0)}(\omega,n_m,T_m)\geq\epsilon_0\}\cap\mathcal{B}_{m-1}^{(l_0)}\right),
\end{align*} 
(where the union in R.H.S. is disjoint) and $\eqref{tmp11}$ follows from the fact that for every $m\geq0$, $\mathbb{P}(\mathcal{B}_{m-1}^{(l_0)}|\mathcal{B}_{-1}^{(l_0)})\leq 1$. Using $\eqref{tmp7}$ and $\eqref{tmp8}$ in $\eqref{tmp11}$, we get that there exists $l_0\geq1$ and $N_0'\geq1$, such that for every $n_0\geq N_0'$, for every $E\in\mathscr{F}_{n_0}$ such that $E\subseteq\{\omega\in\Omega: X_{n_0}(\omega)\in \mathcal{O}'\}$ and $\mathbb{P}(E)>0$,
\begin{align*}
\mathbb{P}\left(X_n\to A\ \text{as }n\to\infty|E\right) &\geq1- \sum_{m=0}^{\infty}\mathbb{P}\left(\{\omega\in\Omega:\rho_1^{(l_0)}(\omega,n_m,T_m)+\rho_2^{(l_0)}         											   (\omega,n_m,T_m)\geq\epsilon_0\}|\mathcal{B}_{m-1}^{(l_0)}\right)\\
                                                                                         &\geq1-\sum_{m=0}^{\infty}\mathbb{P}\left(\{\omega\in\Omega:\rho_1^{(l_0)}(\omega,n_{m},T_m)\geq\frac{\epsilon_0}{2}\}|              								    \mathcal{B}_{m-1}^{(l_0)}\right)\\
                                                                                         &\geq 1-\sum_{m=0}^{\infty}\mathbb{P}\left(\{\omega\in\Omega:\max_{n_m\leq j\leq n_{m+1}}\|\zeta_j(\omega)-\zeta_{n_m}										    (\omega)\|\geq \frac{\epsilon_0}{4K_0(T_u)}\}|\mathcal{B}_{m-1}^{(l_0)}\right).\qed
\end{align*}
\end{proof}
%-----------------------------------------------------------------------------------------------------------------------------------------------------------------------------------------------------------------------------------

\subsection{Review of the probability bounding procedure for the additive noise terms}
\label{nbd}
In this section we fix $l_0$ and $n_0\geq N_0$, where $l_0$ and $N_0$ are as in Lemma \ref{bd1} and provide an upper bound for $\mathbb{P}\left(\{\omega\in\Omega:\max_{n_m\leq j\leq n_{m+1}}\|\zeta_j(\omega)-\zeta_{n_m}(\omega)\|\geq \frac{\epsilon_0}{4K_0(T_u)}\}|\mathcal{B}_{m-1}^{(l_0)}\right)$, for every $m\geq0$. The proof of the bounding procedure is similar to that of \cite[Ch.~4, Lemma.~10]{borkartxt} and we provide a brief outline here for the sake of completeness.
\begin{itemize}
\item [(a)] From recursion $\eqref{rec}$, we have that for every $m\geq0$, for every $n_m\leq j\leq n_{m+1}-1$, for every $\omega\in\Omega$, there exists $V_{j}(\omega)\in F(X_j(\omega))$, such that,
\begin{equation*}
X_{j+1}(\omega)-X_{j}(\omega)-a(n)M_{j+1}(\omega)=a(n)V_{j}(\omega).
\end{equation*}
By assumption $(A1)$(ii), we know that $\|V_j(\omega)\|\leq K(1+\|X_j(\omega)\|)$ and hence for $n_m\leq j\leq n_{m+1}-1$,
\begin{equation*}
\|X_{j+1}(\omega)\|\leq \|X_j(\omega)\|(1+a(j)K)+a(j)K+a(j)\|M_{j+1}\|.
\end{equation*}
Further by assumption $(A3)$, we get that, for every $m\geq0$, for almost every $\omega\in\Omega$, for $n_m\leq j\leq n_{m+1}-1$, 
\begin{equation*}
\|X_{j+1}(\omega)\|\leq \|X_j(\omega)\|(1+2a(j)K)+2a(j)K.
\end{equation*}
Now by arguments as in \cite[Lemma~9]{borkartxt}, we get that, for every $m\geq0$, for almost every $\omega\in\Omega$, for $n_m\leq j\leq n_{m+1}$,
\begin{equation}
\label{tmp13}
\|X_j(\omega)\|\leq e^{2K T_u}(\|X_{n_m}(\omega)\|+2KT_u).
\end{equation}
\item [(b)] Clearly $\{\zeta_j-\zeta_{n_m},\mathscr{F}_j\}_{n_m\leq j\leq n_{m+1}}$ is a martingale. By \eqref{tmp13} and $(A3)$, we get that for $n_m\leq j<n_{m+1}$,
$\|\zeta_{j+1}-\zeta_j\|=\|a(j)M_{j+1}\|\leq a(j)K(1+\|X_j\|)\leq a(j) K(1+e^{2KT_u}(1+2KT_u\|X_{n_m}\|))$. Since for every $\omega\in\mathcal{B}^{(l_0)}_{m-1}$, $X_{n_m}(\omega)\in \mathcal{O}'$ (whose closure is compact), there exists a $C>0$, such that $\|X_{n_m}(\omega)\|\leq C$. Therefore for every $m\geq0$, for every $\omega\in\mathcal{B}^{(l_0)}_{m-1}$, 
for every $n_m\leq j<n_{m+1}$, $\|\zeta_{j+1}-\zeta_j\|\leq a(j)K(1+e^{2KT_u}(1+2KT_uC))$. Thus applying the concentration inequality for martingales, by arguments exactly the same as in the proof of \cite[Lemma~10]{borkartxt}, we get that for every $m\geq0$,
\begin{equation}
\label{tmp14}
\mathbb{P}\left(\{\omega\in\Omega:\max_{n_m\leq j\leq n_{m+1}}\|\zeta_j(\omega)-\zeta_{n_m}(\omega)\|\geq \frac{\epsilon_0}{4K_0(T_u)}\}|\mathcal{B}_{m-1}^{(l_0)}\right)\leq 2de^{-\tilde{K}/(b(n_m)-b(n_{m+1}))}
\end{equation}
where, $\tilde{K}:=\epsilon_0^2/\left(32(K_0(T_u))^2dK(1+e^{2KT_u}(1+2KT_uC))\right)$.
\end{itemize}                               
\begin{proof}

\textbf{of Theorem \ref{mnres}:}
Let $l_0\geq1$ and $N_0'$ be as in Lemma \ref{bd1}. By definition of $b(\cdot)$, we get that there exists $N_0''\geq1$, such that for every $n\geq N_0''$, $b(n)<\tilde{K}$. Define $N_0:=\max\{N_0',N_0''\}$. Let $n_0\geq N_0$ and $\{n_m:=\tau(n_{m-1},T_A)\}_{m\geq1}$. $\{n_m\}_{m\geq1}$ as defined satisfies the conditions mentioned in \ref{d5}. Then by Lemma \ref{bd1} and $\eqref{tmp14}$, we get that for $n_0\geq N_0$, 
\begin{equation}
\label{tmp16}
\mathbb{P}\left(X_n\to A\ \text{as }n\to\infty|E\right)\geq 1-2d\sum_{m=0}^{\infty}e^{-\tilde{K}/(b(n_m)-b(n_{m+1}))}.
\end{equation}
We know that $e^{-\tilde{K}/x}/x\to0$ as $x\to0$ and increases with $x$ for $0<x<\tilde{K}$. Therefore by our choice of $n_0$, we get that, 
\begin{equation*}
\frac{e^{-\tilde{K}/(b(n_m)-b(n_{m+1}))}}{b(n_m)-b(n_{m+1})}\leq\frac{e^{-\tilde{K}/b(n_0)}}{b(n_0)},
\end{equation*}
from which it follows that for every $m\geq0$, $e^{-\frac{\tilde{K}}{b(n_m)-b(n_{m+1})}}\leq (b(n_m)-b(n_{m+1}))\frac{e^{-\frac{\tilde{K}}{b(n_0)}}}{b(n_0)}.$
Substituting the above in $\eqref{tmp16}$, we get that for every $n_0\geq N_0$,
\begin{align*}
\mathbb{P}\left(X_n\to A\ \text{as }n\to\infty|E\right) &\geq 1-2d\sum_{m=0}^{\infty}e^{-\frac{\tilde{K}}{b(n_m)-b(n_{m+1})}}\\
                                                                                        &\geq 1-2d\sum_{m=0}^{\infty} (b(n_m)-b(n_{m+1}))\frac{e^{-\tilde{K}/b(n_0)}}{b(n_0)}\\
                                                                                        &=1-2d\frac{e^{-\tilde{K}/b(n_0)}}{b(n_0)}\sum_{m=0}^{\infty} (b(n_m)-b(n_{m+1}))\\
                                                                                        &=1-2de^{-\tilde{K}/b(n_0)}.\qed
\end{align*}
\end{proof}
\section{Proof of finite resets theorem (Thm.~\ref{fnrst})}
\label{prf_fnrst}

From the definition of $\chi_n$ in equation $\eqref{indi}$, we know that the $\chi_n$ takes the value one if there is a reset of the $n^{th}$ iterate and is zero otherwise. Therefore $\sum_{n=1}^{\infty}\chi_n$ denotes the total number of resets. 

Suppose the event $\{\sum_{n=1}^{\infty}\chi_n\geq k\}$ has zero probability for some $k\geq1$. Then for $k\geq1$, such that $\mathbb{P}(\sum_{n=1}^{\infty}\chi_n\geq k)=0$, we have $\mathbb{P}(\sum_{n=1}^{\infty}\chi_n< k)=1$, from which Theorem \ref{fnrst} trivially follows. Therefore without loss of generality assume $\mathbb{P}(\sum_{n=1}^{\infty}\chi_n\geq k)>0$, for every $k\geq1$.

 For every $k\geq0$, let $G_k$ denote the event that there are at most $k$ resets and $G_{\infty}$ denote the event that there are finitely many resets. That is, for every $k\geq0$, $G_k:=\{\sum_{n=1}^{\infty}\chi_n\leq k\}$ and $G_{\infty}:=\{\sum_{n=1}^{\infty}\chi_n<\infty\}$. Then it is clear that, for every $k\geq 1$, $G_k\subseteq G_{k+1}$ and $G_{\infty}=\cup_{k\geq0}G_{k}$. Therefore $\lim_{k\to\infty}\mathbb{P}(G_k)$ exists and $\mathbb{P}(G_{\infty})=\lim_{k\to\infty}\mathbb{P}(G_k)$. For any $k\geq1$,
 \begin{equation}
 \label{tmp17}
 \mathbb{P}(G_k)=\mathbb{P}(\sum_{n=1}^{\infty}\chi_n\leq k)
                             =\mathbb{P}(\{\sum_{n=1}^{\infty}\chi_n\leq k-1\}\cup\{\sum_{n=1}^{\infty}\chi_n = k\})
                             =\mathbb{P}(G_{k-1})+\mathbb{P}(\sum_{n=1}^{\infty}\chi_n =k).
 \end{equation}
The event $\{\sum_{n=1}^{\infty}\chi_n=k\}$ can be written as a disjoint union of events as below. For every $k\geq1$
\begin{equation}
\label{tmp18}
\left\{\sum_{n=1}^{\infty}\chi_n=k\right\}=\cup_{n_0\geq1}\left[\left\{\sum_{n=1}^{n_0-1}\chi_n=k-1\right\}\cap\left\{\chi_{n_0}=1\right\}\cap\left\{\sum_{n=n_0+1}^{\infty}\chi_n=0\right\}\right],
\end{equation}
where, $\{\sum_{n=1}^{0}\chi_n=k-1\}:=\Omega$. Let $J(k):=\left\{n_0\geq 1: \mathbb{P}\left(\left\{\sum_{n=1}^{n_0-1}\chi_n=k-1\right\}\cap\left\{\chi_{n_0}=1\right\}\right)>0\right\}$. Then for every $k\geq1$,
\begin{itemize}
\item [(a)] By arguments in the second paragraph of this section we have that $\mathbb{P}(G_{k-1}^c)=\mathbb{P}(\sum_{n=1}^{\infty}\chi_n\geq k)>0$. Further the event $\{\sum_{n=1}^{\infty}\chi_n\geq k\}$ can be written as a disjoint union of events as below.
\begin{equation}
\label{tmp19}
\left\{\sum_{n=1}^{\infty}\chi_n\geq k\right\}=\cup_{n_0\geq1}\left[\left\{\sum_{n=1}^{n_0-1}\chi_n=k-1\right\}\cap\left\{\chi_{n_0}=1\right\}\right],
\end{equation}  
from which it follows that,
\begin{equation}
\label{tmp20}
0<\mathbb{P}\left(\left\{\sum_{n=1}^{\infty}\chi_n\geq k\right\}\right)=\sum_{n_0=1}^{\infty} \mathbb{P}\left(\left\{\sum_{n=1}^{n_0-1}\chi_n=k-1\right\}\cap\left\{\chi_{n_0}=1\right\}\right).
\end{equation}
Therefore $J(k)\neq\emptyset$.
\item [(b)] $\min\{n_0\in J(k)\}\geq k$, since there cannot be $k$ resets in less than $k$ iterations.
\end{itemize}
From $\eqref{tmp18}$ and definition of $J(k)$, we have that for every $k\geq1$,
\begin{equation}
\label{tmp21}
\mathbb{P}\left(\sum_{n=1}^{\infty}\chi_n=k\right)=\sum_{n_0\in J(k)}\mathbb{P}\left(\sum_{n=n_0+1}^{\infty}\chi_n=0\bigg{|}\sum_{n=1}^{n_0-1}\chi_n=k-1,\chi_{n_0}=1\right)
\mathbb{P}\left(\sum_{n=1}^{n_0-1}\chi_n=k-1,\chi_{n_0}=1\right).
\end{equation}

\emph{Step 1 \textbf{(Obtaining $\mathcal{O}',\ \epsilon_0$ and $T_A$)} : } By $(A4)'$, we have that $A$ is a globally attracting set of DI $\eqref{mfld}$. Let $\tilde{r}>0$ be such that $A\subseteq \tilde{r}\mathring{U}$. By definition of a globally attracting set and \cite[Lemma~3.13]{benaim1}, we get that for any $r\geq\tilde{r}$, $r\mathring{U}$ is a fundamental neighborhood of $A$. Let $k_1\geq1$ be such that $r_{k_1}\geq \tilde{r}$. Set the fundamental neighborhood $\mathcal{O}:=r_{k_1+1}\mathring{U}$ and $\mathcal{O}':=r_{k_1}\mathring{U}$. Obtain $\epsilon_0>0$ and $T_A>0$ as in section \ref{prf_lip3}. That is $\epsilon_0>0$ is such that $N^{2\epsilon_0}(A)\subseteq \mathcal{O'}\subseteq N^{\epsilon_0}(\bar{\mathcal{O'}})\subseteq \mathcal{O}$ and $T_A>0$, is such that for every $x\in\mathcal{O}$, for every $t\geq T_A$, $\Phi(t,x)\in N^{\epsilon_0}(A)$. 

\emph{Step 2 \textbf{(Obtaining $\{n_m\}_{m\geq1}$ as in \ref{d5})} : } Clearly there exists $k_2\geq1$, such that for every $k\geq k_2$, $T_A\leq2^kT_W$. For any $n_0\geq1$, for every $m\geq1$, define $n_m:=n_{2^{k_2},m-1}$, where for every $1\leq j\leq 2^{k_2}$, $n_{j,m-1}:=\tau(n_{j-1,m-1},T_W)$ with $n_{0,m-1}:=n_{m-1}$ and $\tau(\cdot,\cdot)$ is as defined in \ref{d1}. Therefore for every $m\geq1$, $T_{m-1}:=t(n_m)-t(n_{m-1})=\sum_{j=0}^{2^{k_2}-1}\Delta(n_{j,m-1},T_W)$, where $\Delta(\cdot,\cdot)$ is as defined in \ref{d1}. Thus for every $m\geq0$, $T_A\leq 2^{k_2}T_W\leq T_m\leq 2^{k_2}T_W+2^{k_2}$ and hence $T_u=2^{k_2}T_W+2^{k_2}$.

\emph{Step 3 \textbf{(Redefining trajectories)} : }  Define $\bar{X}$, as defined in $\eqref{lit}$, with the iterates $\{X_n\}_{n\geq0}$ (iterates before reset check) generated by Algorithm \ref{ssri}. For every $n\geq1$, define $\tilde{X}(\cdot,\cdot;n):\Omega\times[t(n),\infty)\rightarrow\mathbb{R}^d$ such that for every $(\omega,t)\in\Omega\times[t(n),t(n+1))$, 
\begin{equation}
\label{rec2a}
\tilde{X}(\omega,t;n):=\left(\frac{t-t(n)}{t(n+1)-t(n)}\right)X_{n+1}(\omega)+\left(\frac{t(n+1)-t}{t(n+1)-t(n)}\right)X'_n(\omega),
\end{equation}
and for every $(\omega,t)\in[t(n+1),\infty)$, $\tilde{X}(\omega,t;n)=\bar{X}(\omega,t)$. 

\emph{Step 4 \textbf{(Obtaining parameters)} : } By arguments exactly same as the ones used to obtain $\eqref{rec1}$, we get that, for every $l\geq1$, for every $n\geq0$, there exists a $U$-valued random variable on $\Omega$, say $\tilde{U}^{(l)}_n$ such that,  for every $\omega\in\Omega$,
\begin{equation}
\label{rec2}
X_{n+1}(\omega)-X_{n}'(\omega)-a(n)M_{n+1}(\omega)= a(n)f^{(l)}(X_n'(\omega),\tilde{U}^{(l)}_n(\omega)).
\end{equation}

\emph{Step 5 \textbf{(Redefining distance measures)} : }
For every $\omega\in\Omega$, for every $n\geq n'\geq1$, for every $T>0$, for every $l\geq1$,
\begin{itemize}
\item [(a)] let $\tilde{\bm{\mathrm{x}}}^{(l)}(\cdot;n,n',T,\omega):[0,T]\rightarrow\mathbb{R}^d$ denote the unique solution of the o.d.e. 
\begin{equation}
\label{pmfld2}
\frac{dx}{dt}=f^{(l)}(x,\tilde{u}(t;n,n',T,\omega)),
\end{equation}
with initial condition $\tilde{\bm{\mathrm{x}}}^{(l)}(0;n,n',T,\omega)=X_{n'}'(\omega)$, where $\tilde{u}(\cdot;n,n',T,\omega):[0,T]\rightarrow U$ is defined such that, for every $t\in[0,T]$, $\tilde{u}(t;n,n',T,\omega):=\tilde{U}_k^{(l)}(\omega)$, where $\tilde{U}_k^{(l)}$ is as in equation $\eqref{rec2}$ and $k$ is such that $t+t(n)\in [t(k),t(k+1))$ (for a proof of existence and uniqueness of solutions to o.d.e. $\eqref{pmfld2}$, see Lemma \ref{pvf}).  It is easy to see that for every $l\geq1$, $\tilde{\bm{\mathrm{x}}}^{(l)}(\cdot;n,n',T,\omega)\in S^{(l)}(T,X_{n'}'(\omega))$, the set of solutions of DI $\eqref{amfld}$, as defined in $\eqref{ss1}$.
\item [(b)] define,
\begin{itemize} 
\item [(1)] $\tilde{\rho}(\omega,n,n',T):=\inf_{\bm{\mathrm{x}}\in S(T,\bar{\mathcal{O}'})}\sup_{t\in[0,T]}\|\tilde{X}(\omega,t+t(n);n')-\bm{\mathrm{x}}(t)\|$,  
\item [(2)]$\tilde{\rho}^{(l)}_1(\omega,n,n',T):=\sup_{t\in[0,T]}\|\tilde{X}(\omega,t+t(n);n')-\tilde{\bm{\mathrm{x}}}^{(l)}(t;n,n',T,\omega)\|$, 
\item [(3)]$\tilde{\rho}^{(l)}_2(\omega,n,n',T):=\inf_{\bm{\mathrm{x}}\in S(T,\bar{\mathcal{O}'})}\sup_{t\in[0,T]}\|\tilde{\bm{\mathrm{x}}}^{(l)}(t;n,n',T,\omega)-\bm{\mathrm{x}}(t)\|$. 
\end{itemize}
\end{itemize}

\emph{Step 6 \textbf{(Collecting sample paths)} : } Fix $k> \max\{k_1,k_2\}$ and $n_0\in J(k)$. By our definition of $\mathscr{F}_{n_0}$ (see section \ref{appl}), we have that $E(k,n_0):=\{\sum_{n=1}^{n_0-1}\chi_n=k-1,\chi_{n_0}=1\}\in\mathscr{F}_{n_0}$ and is contained in $\{X_{n_0}'(\omega)\in\mathcal{O}'\}$. Given that there has been a reset at index $n_0$, the next reset check is performed by Algorithm \ref{ssri} at the iteration index $n_{2^{k-k_2}}$. So for $n_0+1\leq j< n_{2^{k-k_2}}$, $X_j(\omega)=X_j'(\omega)$. From arguments exactly the same as Lemma \ref{events}$(a)$, we get that,  
\begin{align}
E(k,n_0)\cap\big(\cap_{m=0}^{2^{k-k_2}-1}\{\tilde{\rho}^{(l)}_1(\omega,n_m,n_0,T_m)+\tilde{\rho}^{(l)}_2&(\omega,n_m,n_0,T_m)<\epsilon_0\}\big)\nonumber\\
&\subseteq E(k,n_0)\cap\left(\cap_{m=0}^{2^{k-k_2}-1}\{\tilde{\rho}(\omega,n_m,n_0,T_m)<\epsilon_0\}\right)\nonumber\\
&\subseteq\{\omega\in\Omega:X_{n_{2^{k-k_2}}}(\omega)\in\mathcal{O}'\}\nonumber \\
\label{tmp22}
&\subseteq \{\omega\in\Omega: X_{n_{2^{k-k_2}}}(\omega)=X_{n_{2^{k-k_2}}}'(\omega)\}\\
&\subseteq \{\sum_{n=n_0+1}^{n_{2^{k-k_2}}}\chi_n=0\}\nonumber
\end{align}
where, $\eqref{tmp22}$ follows from the fact that $k\geq k_1$ and hence $\mathcal{O}'= r_{k_1}\mathring{U}\subseteq r_{k}\mathring{U}$. It is also worth mentioning here that the proof of Lemma \ref{events}$(a)$ holds irrespective of how the iterates are generated. Given that $\omega\in E(k,n_0)\cap\left(\cap_{m=0}^{2^{k-k_2}-1}\{\tilde{\rho}^{(l)}_1(\omega,n_m,n_0,T_m)+\tilde{\rho}^{(l)}_2(\omega,n_m,n_0,T_m)<\epsilon_0\}\right)$, along this sample path there has been a reset at $n_0$ and at the next check performed at $n_{2^{k-k_2}}$ there has been no reset. Hence the next check for reset is performed by Algorithm \ref{ssri} at $n_{2^{(k-k_2)+1}}$. Again from arguments from Lemma \ref{events}$(a)$, we get that,
\begin{align*}
E(k,n_0)\cap\left(\cap_{m=0}^{2^{(k-k_2)+1}-1}\{\tilde{\rho}^{(l)}_1(\omega,n_m,n_0,T_m)+\tilde{\rho}^{(l)}_2(\omega,n_m,n_0,T_m)<\epsilon_0\}\right)&\subseteq\{\omega\in\Omega:X_{2^{(k-k_2)+1}}(\omega)\in\mathcal{O}'\}\\
&\subseteq \{\sum_{n=n_0+1}^{n_{2^{(k-k_2)+1}}}\chi_n=0\}.
\end{align*}
Repeating the above for the third reset check after $n_0$ and so on, we obtain that,
\begin{equation}
\label{tmp23}
E(k,n_0)\cap\left(\cap_{m\geq0}\{\tilde{\rho}^{(l)}_1(\omega,n_m,n_0,T_m)+\tilde{\rho}^{(l)}_2(\omega,n_m,n_0,T_m)<\epsilon_0\}\right)
\subseteq \{\sum_{n=n_0+1}^{\infty}\chi_n=0\}.
\end{equation}

\emph{Step 7 \textbf{(Bounding)} : } Define $\tilde{\mathcal{B}}^{(l)}_{-1}:=E(k,n_0)$ and for every $M\geq1$ define, 
\begin{equation*}
\tilde{\mathcal{B}}^{(l)}_M:=E(k,n_0)\cap\left(\cap_{m=0}^{M}\{\tilde{\rho}^{(l)}_1(\cdot,n_m,n_0,T_m)+\rho^{(l)}_2(\cdot,n_m,n_0,T_m)<\epsilon_0\}\right).
\end{equation*}
Note that as in Lemma \ref{bd1}, we can obtain $l_0\geq1$, such that for every $m\geq0$, for every $\omega\in\tilde{B}^{(l_0)}_{m-1}$ we have that $\tilde{\rho}^{(l_0)}_2(\omega,n_m,n_0,T_m)<\frac{\epsilon_0}{2}$, since for every $\omega\in\tilde{\mathcal{B}}^{(l_0)}_{m-1}$, $X_{n_m}(\omega)\in\mathcal{O}'$ and whether or not a reset check is performed at this index, we have that $X_{n_m}(\omega)=X_{n_m}'(\omega)$. Thus for such an $l_0$, mimicking the proof of Lemma \ref{bd1}, we obtain that,
\begin{equation*}
\mathbb{P}\left(\cap_{m=0}^{M}\{\tilde{\rho}^{(l)}_1(\cdot,n_m,n_0,T_m)+\rho^{(l)}_2(\cdot,n_m,n_0,T_m)<\epsilon_0\}|\tilde{\mathcal{B}}^{(l_0)}_{-1}\right)\geq 1-\sum_{m=0}^{\infty}\mathbb{P}\left(\tilde{\rho}^{(l_0)}_1(\cdot,n_m,n_0,T_m)\geq\frac{\epsilon_0}{2}|\tilde{\mathcal{B}}^{(l_0)}_{m-1}\right).
\end{equation*}

From Lemma \ref{bd0}, we have that for $k> \max\{k_1,k_2,N_0'\}$, for every $n_0\in J(k)$, 
\begin{small}
\begin{equation}
\label{tmp24}
  \mathbb{P}\!\left(\!\cap_{m=0}^{M}\{\tilde{\rho}^{(l)}_1(\cdot,n_m,n_0,T_m)+\rho^{(l)}_2(\cdot,n_m,n_0,T_m)\!<\epsilon_0\}|\tilde{\mathcal{B}}^{(l_0)}_{-1}\!\right)\!\geq 1-\sum_{m=0}^{\infty}\mathbb{P}\!\left(\max_{n_m\leq j\leq n_{m+1}}\!\!\!\!\|\zeta_j-\zeta_{n_m}\|\geq \frac{\epsilon_0}{4K_0(T_u)}|\tilde{\mathcal{B}}^{(l_0)}_{m-1}\!\right).
\end{equation}
\end{small}

\emph{Step 8 \textbf{(Noise bound)}} Similar to item $(a)$ in section \ref{nbd}, from Algorithm \eqref{ssri}, we have that for every $m\geq0$, for every $n_m\leq n_{m+1}-1$, $\|X_{j+1}\|\leq \|X_j'\|(1+a(j)K)+a(j)K+a(j)\|M_{j+1}\|$ and since $\|X_{j+1}'\|\leq\|X_{j+1}\|$, we get that for every $n_m\leq j\leq n_{m+1}-1$,
\begin{equation*}
\|X'_{j+1}\|\leq \|X_j'\|(1+a(j)K)+a(j)K+a(j)\|M_{j+1}\|.
\end{equation*}
Now by arguments exactly same as those item $(a)$ of section \ref{nbd}, we get that for every $m\geq0$, for every $n_m\leq j\leq n_{m+1}-1$, $\|X_{j+1}'\|\leq e^{2K T_u}(\|X_{n_m}\|+2KT_u)$. Now by using concentration inequality as in item $(b)$ of section \ref{nbd}, we get that for every $m\geq0$,
\begin{equation}
\label{tmp25}
\mathbb{P}(\max_{n_m\leq j\leq n_{m+1}}\|\zeta_j-\zeta_{n_m}\|\geq \frac{\epsilon_0}{4K_0(T_u)}|\tilde{\mathcal{B}}_{m-1}^{(l_0)})\leq2de^{-\tilde{K}/(b(n_m)-b(n_{m+1}))}.
\end{equation} 

Using $\eqref{tmp23}$, $\eqref{tmp24}$ and $\eqref{tmp25}$ we get that, for every $k\geq \max\{k_1,k_2,N_0\}$ (where $N_0$ is as defined in the proof of Theorem \ref{mnres}), for every $n_0\in J(k)$, 
\begin{equation*}
\mathbb{P}\left(\sum_{n=n_0+1}^{\infty}\chi_n=0\bigg{|}\sum_{n=1}^{n_0-1}\chi_n=k-1,\chi_{n_0}=1\right)\geq 1-2d\sum_{m=0}^{\infty}e^{-\tilde{K}/(b(n_m)-b(n_{m+1}))}\geq1-2de^{-\tilde{K}/b(n_0)}. 
\end{equation*} 
Substituting the above in $\eqref{tmp21}$ and using the fact that for $n\leq n'$, $b(n')\leq b(n)$, we get that, for every $k\geq\max\{k_1,k_2,N_0\}$,
\begin{align}
 \mathbb{P}\left(\sum_{n=1}^{\infty}\chi_n=k\right)&\geq(1-2de^{-\tilde{K}/b(k)})\sum_{n_0=1}^{\infty} \mathbb{P}\left(\left\{\sum_{n=1}^{n_0-1}\chi_n=k-1\right\}\cap\left\{\chi_{n_0}=1\right\}\right)\nonumber\\
 \label{tmp26}
 &=(1-2de^{-\tilde{K}/b(k)})\mathbb{P}\left(\left\{\sum_{n=1}^{\infty}\chi_n\geq k\right\}\right)
 \end{align}
 Substituting $\eqref{tmp26}$ in $\eqref{tmp17}$, we get that for every $k\geq\max\{k_1,k_2,N_0\}$,
 \begin{equation*}
 \mathbb{P}(G_k)\geq \mathbb{P}(G_{k-1})+(1-2de^{-\tilde{K}/b(k)})\mathbb{P}(G_{k-1}^c)
                            \geq 1-2de^{-\tilde{K}/b(k)}.
\end{equation*}
Letting $k\to\infty$ in the above equation and using the fact that $\mathbb{P}(G_{\infty})=\lim_{k\to\infty}\mathbb{P}(G_k)$, we get that $\mathbb{P}(G_{\infty})=1$.
\section{Conclusions and directions for future work}
\label{concl}
We have extended the lock-in probability result (Theorem \ref{mnres}) in \cite{borkarlkp} to stochastic approximation schemes with set-valued drift functions which serves as an important tool for analyzing recursions when their stability is not guaranteed. The extension to set-valued map allows one to obtain lock-in probability for stochastic approximation schemes with measurable drift functions and schemes where the drift function itself possess a non-additive unknown noise component (see \cite[Ch.~5.3]{borkartxt}). Further using Theorem \ref{mnres}, in the presence of a locally attracting set for the mean field, we have provided an alternate condition for verification of convergence in the absence of stability guarantee which involves verifying whether the iterates are entering infinitely often, an open neighborhood of the attractor with a compact closure.  In the presence of a globally attracting set our modified recursion as in Algorithm \ref{ssri}, converges almost surely to the globally attracting set, the proof of which relies on the method used to obtain the lock-in probability result.
 
In future we wish to consider other applications of the lock-in probability result such as sample complexity (see \cite[Ch.~4.2]{borkartxt}) and almost sure convergence under tightness of the iterates (see \cite{kamal}). Another interesting direction, is to explore various additive noise models where the above result can be extended for the case of set-valued drift functions.

\appendix

\section{Definitions of some topological concepts}
\label{topo}
Let $(\mathcal{M},\Gamma)$ be a topological space. $\{O_i\}_{i\in I}$ is an \it{covering }\rm of $\mathcal{M}$ if, for every $i\in I$, $O_i\subseteq\mathcal{M}$ and $\cup_{i\in I}O_i=\mathcal{M}$. Further a covering $\{O_i\}_{i\in I}$ is said to be \it{locally finite }\rm if for every $p\in\mathcal{M}$, there exists an $O\in \Gamma$, such $O_i\cap O\neq\emptyset$ for only finitely many $i\in I$. Given any two coverings $\mathcal{C}:=\{O_i\}_{i\in I}$ and $\mathcal{C}':=\{O_j'\}_{j\in J}$, $\mathcal{C}'$ is said to be a \it{refinement }\rm of $\mathcal{C}$ if, for every $i\in I$, there exists a $j\in J$ such that, $O_i\subseteq O_j'$. $\mathcal{C}'$ is said to be a locally finite refinement of $\mathcal{C}$, if $\mathcal{C}'$ is a refinement of $\mathcal{C}$ and is locally finite. The topological space $(\mathcal{M},\Gamma)$ is \it{paracompact }\rm if it is a Hausdorff space and if every open covering has a locally finite open refinement.

A family of functions $\{\psi_i\}_{i\in I}$ is called a \it{locally Lipschitz partition of unity }\rm if for all $i\in I$,
\begin{itemize}
\item $\psi_i$ is locally Lipschitz continuous and non negative,
\item the supports of $\psi_i$, defined as $\overline{\{p\in \mathcal{M}:\psi_i(p)\neq0\}}$, are a closed locally finite covering of $\mathcal{M}$,
\item for each $p\in\mathcal{M}$, $\sum_{i\in I}\psi_i(p)=1$.
\end{itemize}
A partition of unity $\{\psi_i\}_{i\in I}$ is said to be subordinated to the covering $\{O_i\}_{i\in I}$, if for every $i\in I$, $\overline{\{p\in \mathcal{M}:\psi_i(p)\neq0\}}\subseteq O_i$.

\section{Proof of Lemma \ref{bd0}}
\label{prf_int} 
Fix $l\geq1$ and $T_u\geq T_A$. Fix $n_0\geq 1$ and $\{n_m\}_{m\geq0}$ as in \ref{d5}. Fix $m\geq0$. Let $\omega\in\{\omega\in\Omega:\rho_1^{(l)}(\omega,n_m,T_m)\geq\frac{\epsilon_0}{2}\}\cap\mathcal{B}_{m-1}^{(l)}$. Then, by Lemma \ref{events}$(a)$ we have that, $X_{n_m}(\omega)\in\mathcal{O}'$. By $\eqref{lit}$ and \ref{d3}, we have that for every $0\leq t\leq T_m$, there exists $n_m\leq k\leq n_{m+1}-1$ such that $t\in [t(k),t(k+1)]$, and 
\begin{equation}
\label{t1}
\bar{X}(\omega,t+t(n_m))=\alpha X_{n_k}(\omega)+(1-\alpha)X_{n_{k+1}}(\omega),
\end{equation} 
for some $\alpha\in [0,1]$ and
\begin{equation}
\label{t2}
\tilde{x}^{(l)}(t;n_m,T_m,\omega)= X_{n_m}(\omega)+\int_{0}^{q}f^{(l)}(\tilde{x}^{(l)}(q;n_m,T_m,\omega),u(q;n_m,T_m,\omega))dq. 
\end{equation}
Therefore for any $t\in [0,T_m]$, 
\begin{align}
\label{t3}
\|\bar{X}(\omega,t+t(n_m))-\tilde{x}^{(l)}(t;n_m,T_m,\omega)\|&\leq \alpha\|X_k(\omega)-\tilde{X}^{(l)}(t(k);n_m,T_m,\omega)\|\nonumber\\& +\alpha \|\tilde{X}^{(l)}(t(k);n_m,T_m,\omega)-\tilde{X}^{(l)}(t;n_m,T_m,\omega)\|\nonumber\\
&+(1-\alpha)\|\tilde{X}^{(l)}(t(k+1);n_m,T_m,\omega)-\tilde{X}^{(l)}(t;n_m,T_m,\omega)\|\nonumber\\&+(1-\alpha)\|X_{k+1}(\omega)-\tilde{X}^{(l)}(t(k+1);n_m,T_m,\omega)\|
\end{align}
Now the aim is to provide an upper bound for each terms in the $R.H.S.$ of the above inequality which is independent of $m$ and $t$. In order to apply \cite[Ch.~2, Lemma~1]{borkartxt} the only additional condition needed is the Lipschitz continuity of $f^{(l)}(\cdot,\cdot)$ uniformly over $u$ and this is obtained using local Lipschitz continuity as follows. Let $M_1>0$ be such that, for every $x\in\mathcal{O}'$, $\|x\|\leq M_1$. By Lemma \ref{pvf}$(b)$ and item (a) in section \ref{nbd}, we have that for $r:=\max\{C_1(\bar{\mathcal{O}'},T_u,l), e^{2KT_u}(M_1+2KT_u)\}$, 
\begin{equation*}
\sup_{t\in[0,T_m]}\|\bar{X}(\omega,t+t(n_m))\|\leq r\ \mathrm{and}\ \sup_{t\in [0,T_u]}\|\tilde{x}^{(l)}(t;n_m,T_m,\omega)\|\leq r.
\end{equation*}
Further by Lemma \ref{pvf}$(c)$, we know that there exists $L(r,l)>0$, such that for every $x,x'\in r U$, for every $t\in [0, T_m]$
\begin{equation*}
\|f^{(l)}(x,u(t;n_m,T_m,\omega))-f^{(l)}(x',u(t;n_m,T_m,\omega))\|\leq L(r,l)\|x-x'\|.
\end{equation*}
The rest of the bounding procedure is exactly the same as \cite[Ch.~2, Lemma~1]{borkartxt}. We obtain that, for every $t\in [0,T_m]$, 
\begin{align*}
\sup_{t\in [0,T_m]}\|\bar{X}(\omega,t+t(n_m))-\tilde{x}^{(l)}(t;n_m,T_m,\omega)\|\leq (&M_1+KT_u)e^{2L(r,l)T_u}L(r,l)\sum_{j\geq0}a(n_0+j)^2\\
&+e^{L(r,l)T_u}\max_{n_m\leq j\leq n_{m+1}}\|\zeta_j(\omega)-\zeta_{n_m}(\omega)\|\\
&+(M_1+KT_u)e^{L(r,l)T_u}a(n_0).
\end{align*}

Now set $N_0'$ such that $(M_1+KT_u)e^{2L(r,l)T_u}L(r,l)\sum_{j\geq0}a(n_0+j)^2+(M_1+KT_u)e^{L(r,l)T_u}a(n_0)<\frac{\epsilon_0}{2}$ and define $K_0(T_u):=e^{L(r,l)T_u}$. Then, for every $n_0\geq N_0'$, we have that $\omega\in\{\omega\in\Omega:\max_{n_m\leq j\leq n_{m+1}}\|\zeta_j(\omega)-\zeta_{n_m}(\omega)\|\geq \frac{\epsilon_0}{4K_0(T_u)}\}\cap\mathcal{B}^{(l)}_{m-1}$, from which Lemma \ref{bd0} follows.
\bibliographystyle{IEEEtran}
\bibliography{Ref}
\end{document}